\documentclass[11pt]{article} 
\usepackage{fullpage}
\usepackage{graphicx}
\usepackage{latexsym}

\usepackage{multirow}
\usepackage{mdframed}
\usepackage{tikz}
\usepackage{color} 
\usepackage[linesnumbered,ruled]{algorithm2e} 

\usepackage[toc,page]{appendix}

\usepackage{amsmath, amsthm, amssymb}           
\usepackage[english]{babel}                                  
\usepackage[T1]{fontenc}                                      
\usepackage[utf8]{inputenc}
\usepackage[round]{natbib}
\usepackage{array}           
\usepackage{geometry}
\usepackage{bbm}
\usepackage{hyperref}
\usepackage{subcaption}
\usepackage{dsfont}
\usepackage{mathtools} 
\usepackage{comment}
\usepackage{stackrel}
\usepackage{pdfpages}
\usepackage{extarrows} 
\usepackage{listings} 
\usepackage{ulem} 
\usepackage{hyphenat} 
\usepackage[toc,page]{appendix}
\usepackage{array}
\newcolumntype{P}[1]{>{\centering\arraybackslash}p{#1}}
\newcolumntype{M}[1]{>{\centering\arraybackslash}m{#1}}
\newcolumntype{L}[1]{>{\raggedright\arraybackslash}m{#1}}
\newcolumntype{C}[1]{>{\centering\arraybackslash}m{#1}}
\newcolumntype{R}[1]{>{\raggedleft\arraybackslash}m{#1}}

\DeclarePairedDelimiter\ceil{\lceil}{\rceil}
\DeclarePairedDelimiter\floor{\lfloor}{\rfloor}

\definecolor{upmaroon}{rgb}{0.48, 0.07, 0.07}
\definecolor{royalazure}{rgb}{0.0, 0.22, 0.66}
\definecolor{pakistangreen}{rgb}{0.0, 0.4, 0.0}

\theoremstyle{plain}                                
\newtheorem{theorem}{Theorem}[section]
\newtheorem{definition}[theorem]{Definition}

\newtheorem{lemma}[theorem]{Lemma}
\newtheorem{assumptions}[theorem]{Assumption}

\usepackage{enumitem}
\setlist[itemize]{noitemsep, topsep=0pt}

\def\iid{independent and identically distributed}

\newcommand{\eps}{\varepsilon}

\newcommand{\R}{\mathbb{R}} 
\newcommand{\bo}{\boldsymbol{0}}
\newcommand{\one}{\mathds{1}}
\newcommand{\oneni}{\mathds{1}_{\nni}}
\newcommand{\lambdamin}{\lambda_{\mathrm{min}}}

\newcommand{\clambdamin}{c_{\mathrm{min}}}

\newcommand\independent{\protect\mathpalette{\protect\independenT}{\perp}}
\def\independenT#1#2{\mathrel{\rlap{$#1#2$}\mkern2mu{#1#2}}}

\newcommand{\norm}[1]{\lVert #1 \rVert}
\newcommand{\normbig}[1]{\big\lVert #1 \big\rVert}
\newcommand{\normBig}[1]{\Big\lVert #1 \Big\rVert}
\newcommand{\normbigg}[1]{\bigg\lVert #1 \bigg\rVert}
\newcommand{\normP}[2]{\norm{#1}_{\PP, #2}}
\newcommand{\normPbig}[2]{\normbig{#1}_{\PP, #2}}
\newcommand{\normPBig}[2]{\normBig{#1}_{\PP, #2}}

\newcommand{\normPNbig}[2]{\normbig{#1}_{\PPN, #2}}

\newcommand{\normone}[1]{\lvert #1\rvert}

\DeclareMathOperator*{\argmin}{arg\,min}
\DeclareMathOperator{\E}{\mathbb{E}}

\DeclareMathOperator{\Prob}{\mathbb{P}}

\newcommand{\hbRkXi}{\widehat{\mathbf{R}}_{\Xi}^{\Ik}}

\newcommand{\bRkXi}{\mathbf{R}_{\Xi}}

\newcommand{\hbRkYi}{\widehat{\mathbf{R}}_{\Yi}^{\Ik}}

\newcommand{\RIk}{\widehat{\mathbf{R}}^{\Ik}}

\newcommand{\bRkYi}{\mathbf{R}_{\Yi}}

\newcommand{\NN}{N}
\newcommand{\nn}{n}
\newcommand{\KK}{K}
\newcommand{\kk}{k}
\newcommand{\CpnormRV}{C_1}
\newcommand{\cone}{c_1}
\newcommand{\ctwo}{c_2}
\newcommand{\CpnormEta}{C_4}
\newcommand{\CboundV}{C_3}
\newcommand{\CboundZi}{C_2}

\newcommand{\TauN}{\mathcal{T}} 
\newcommand{\EpsN}{\mathcal{E}_{\NN}}
\newcommand{\PP}{P}
\newcommand{\PcalN}{\mathcal{P}_{\NN}}

\newcommand{\rNp}{r_{\NN, \kk}'}

\newcommand{\errorN}{e_{\NN}}

\newcommand{\deltaN}{\delta_{\NN}}
\newcommand{\DeltaN}{\Delta_{\NN}}

\newcommand{\PPN}{\PP_{\NN}}

\newcommand{\Enk}[1]{\E_{\nnktot}[#1]}
\newcommand{\Enkbig}[1]{\E_{\nnktot}\big[#1\big]}

\newcommand{\Gnk}{\mathbb{G}_{\PP, \Ik}}
\newcommand{\Gnkbig}[1]{\mathbb{G}_{\PP, \Ik}\big[#1\big]}

\newcommand{\EP}{\E_{\PP}}
\newcommand{\EPbig}[1]{\E_{\PP}\big[#1\big]}
\newcommand{\EPBig}[1]{\E_{\PP}\Big[#1\Big]}

\newcommand{\EPNbig}[1]{\E_{\PPN}\big[#1\big]}

\newcommand{\EPNbigg}[1]{\E_{\PPN}\bigg[#1\bigg]}

\newcommand{\indset}[1]{[#1]}

\newcommand{\SIkc}{\textbf{S}_{\Ikc}}

\newcommand{\etazero}{\eta^0}
\newcommand{\hetaIkc}{\hat\eta^{\Ikc}}

\newcommand{\I}{I}

\newcommand{\Ik}{I_{\kk}}
\newcommand{\Ikc}{I_{\kk}^c}

\newcommand{\Icalone}{\mathcal{I}_1}
\newcommand{\Icaltwo}{\mathcal{I}_2}

\newcommand{\Ti}{T_{\NN,i}}
\newcommand{\Tbar}{\overline{T}_{\NN}}
\newcommand{\nuN}{\nu_{\NN}}

\newcommand{\betazero}{\beta_0}
\newcommand{\Sigmazero}{\Sigma_{0}}
\newcommand{\Gammazero}{\Gamma_{0}}
\newcommand{\Sigmakj}{\Sigma_{\kappa,\iota}}
\newcommand{\sigmazero}{\sigma_0}
\newcommand{\thetazero}{\theta_0}

\newcommand{\Vzeroi}{\mathbf{V}_{0, i}}
\newcommand{\Vione}{\mathbf{V}_{i,1}}
\newcommand{\Vitwo}{\mathbf{V}_{i,2}}
\newcommand{\Vi}{\mathbf{V}_{i}}

\newcommand{\Vizero}{\mathbf{V}_{i, 0}}

\newcommand{\hVik}{\hat{\mathbf{V}}_{i,\kk}}

\newcommand{\li}{\ell_i}
\newcommand{\libig}[1]{\ell_i\big(#1\big)}

\newcommand{\hbeta}{\hat\beta}
\newcommand{\hbetak}{\hat\beta_{\kk}}
\newcommand{\htheta}{\hat\theta}
\newcommand{\hthetak}{\hat\theta_{\kk}}

\newcommand{\hsigmak}{\hat\sigma_{\kk}}

\newcommand{\hSigmak}{\hat{\Sigma}_{\kk}}

\newcommand{\Jzero}{J_0}
\newcommand{\Tzero}{T_0}
\newcommand{\hTzero}{\hat T_0}

\newcommand{\vvv}{v}
\newcommand{\blambda}{\boldsymbol{\lambda}}

\newcommand{\ttt}{u}

\newcommand{\Salg}{\mathcal{S}}
\newcommand{\hthetaks}{\hat\theta_{\kk, s}}
\newcommand{\hbetaks}{\hat\beta_{\kk, s}}
\newcommand{\hsigmaks}{\hat\sigma_{\kk, s}}
\newcommand{\hSigmaks}{\hat\Sigma_{\kk, s}}
\newcommand{\hbetas}{\hat\beta_{s}}
\newcommand{\hTzeros}{\hat T_{0,s}}

\newcommand{\NNT}{\NN_T}
\newcommand{\NNtot}{\NN_{T}}
\newcommand{\nni}{n_i}
\newcommand{\nmax}{n_{\mathrm{max}}}
\newcommand{\nnktot}{\nn_{T, k}}

\newcommand{\Fcaleta}{\mathcal{F}_{\eta}}

\newcommand{\Fcaletazero}{\mathcal{F}_{\etazero}}
\newcommand{\FcalhetaIkc}{\mathcal{F}_{\hetaIkc}}
\newcommand{\Fcaltwo}{\mathcal{F}_{2}}
\newcommand{\Feta}{{F}_{\eta}}
\newcommand{\Fetazero}{{F}_{\etazero}}
\newcommand{\FhetaIkc}{{F}_{\hetaIkc}}
\newcommand{\Ftwo}{{F}_{2}}

\newcommand{\g}{g}

\newcommand{\h}{\zeta}
\newcommand{\scoretest}{\varphi}
\newcommand{\scoreX}{\xi}
\newcommand{\mX}{m_X}
\newcommand{\hmX}{\hat m_X}
\newcommand{\mY}{m_Y}
\newcommand{\hmY}{\hat m_Y}
\newcommand{\psibig}[1]{\psi\big(#1\big)}
\newcommand{\psibeta}{\psi_{\beta}}
\newcommand{\psibetabig}[1]{\psi_{\beta}\big(#1\big)}
\newcommand{\psisigma}{\psi_{\sigma^2}}
\newcommand{\psisigmabig}[1]{\psi_{\sigma^2}\big(#1\big)}
\newcommand{\psiSigmakj}{\psi_{\Sigma_{\kappa,\iota}}}
\newcommand{\psiSigmakjbig}[1]{\psi_{\Sigma_{\kappa,\iota}}\big(#1\big)}

\newcommand{\Yi}{\mathbf{Y}_i}
\newcommand{\Yj}{\mathbf{Y}_j}

\renewcommand{\Xi}{\mathbf{X}_i}
\newcommand{\Xj}{\mathbf{X}_j}
\newcommand{\Zi}{\mathbf{Z}_i}
\newcommand{\Zj}{\mathbf{Z}_j}
\newcommand{\Ztili}{\widetilde{\mathbf{Z}}_i}
\newcommand{\Wi}{\mathbf{W}_i}
\newcommand{\Wj}{\mathbf{W}_j}
\newcommand{\Si}{\mathbf{S}_i}
\newcommand{\Sj}{\mathbf{S}_j}
\newcommand{\si}{\mathbf{s}_i}

\renewcommand{\S}{\mathbf{S}}

\newcommand{\Ubi}{\mathbf{A}_i}
\newcommand{\Vbi}{\mathbf{B}_i}
\newcommand{\Di}{\mathbf{D}_i}

\newcommand{\RYi}{\mathbf{R}_{\Yi}}
\newcommand{\RYj}{\mathbf{R}_{\Yj}}
\newcommand{\RYieta}{\mathbf{R}_{\Yi, \eta}}
\newcommand{\RXi}{\mathbf{R}_{\Xi}}
\newcommand{\RXieta}{\mathbf{R}_{\Xi, \eta}}

\newcommand{\epsi}{\boldsymbol{\varepsilon}_i}
\newcommand{\epsj}{\boldsymbol{\varepsilon}_j}
\newcommand{\epsXi}{\boldsymbol{\varepsilon}_{\Xi}}
\newcommand{\epsXj}{\boldsymbol{\varepsilon}_{\Xj}}
\newcommand{\epsone}{\boldsymbol{\varepsilon}_1}
\newcommand{\epsN}{\boldsymbol{\varepsilon}_{\NN}}

\newcommand{\bbi}{\mathbf{b}_i}
\newcommand{\bj}{\mathbf{b}_j}
\newcommand{\bone}{\mathbf{b}_1}
\newcommand{\bN}{\mathbf{b}_{\NN}}

\begin{document}

\title{Double Machine Learning for Partially Linear Mixed-Effects Models with Repeated Measurements}  
  
\author{Corinne Emmenegger and Peter B\"uhlmann\\
Seminar for Statistics, ETH Z\"urich}

\maketitle

\begin{abstract}
Traditionally, spline or kernel approaches 
in combination with parametric estimation are used to infer
the linear coefficient (fixed effects) in a partially linear 
mixed-effects model for repeated measurements. 
Using machine learning algorithms allows us to incorporate complex interaction structures and high-dimensional variables. 
We employ double machine learning to cope with the nonparametric part of the 
partially linear mixed-effects model: the nonlinear variables are regressed out nonparametrically from both the linear variables and the response. 
This adjustment can be performed with any 
machine learning algorithm, for instance random forests, which allows to take complex interaction terms and nonsmooth structures into account. 
The adjusted variables satisfy a linear mixed-effects model, where the linear coefficient can be estimated with standard linear mixed-effects techniques. 
We prove that the estimated fixed effects coefficient converges at the parametric rate,  is
asymptotically Gaussian distributed, and semiparametrically efficient. 
Two simulation studies demonstrate that our method outperforms a penalized regression spline approach in terms of coverage. 
We also illustrate our proposed approach on a longitudinal dataset with HIV-infected individuals. Software code for our method is available in the \textsf{R}-package \texttt{dmlalg}.
\end{abstract}


\section{Introduction}

Repeated measurements data consists of observations from several experimental units, subjects, or groups under different conditions. 
This grouping or clustering of the individual responses into experimental units 
typically introduces dependencies: 
the different units are assumed to be independent, 
but there may be heterogeneity between units 
and correlation within units.

Mixed-effects models 
provide a powerful and flexible tool to analyze grouped data by incorporating fixed and random effects. Fixed effects are associated with the entire population, and random effects are associated with individual 
groups
and model the heterogeneity across them and the dependence structure within them~\citep{Pinheiro2000}. 
Linear mixed-effects models~\citep{Laird-Ware1982, Pinheiro2000, Verbeke-Molenberghs2009, Demidenko2004} 
 impose a linear relationship between all covariates and the response. 
Partially linear mixed-effects models~\citep{Zeger-Diggle1994}  extend the linear ones. 

We consider the partially linear mixed-effects model 
\begin{equation}\label{eq:initPLMM}
	\Yi = \Xi\betazero + \g(\Wi) + \Zi\bbi+ \epsi
\end{equation}
for groups $i\in\{1,\ldots,\NN\}$. There are $\nni$ observations per group $i$. 
The unobserved random variable $\bbi$, called random effect, introduces 
correlation
within its group $i$ because all $\nni$ observations within this group share $\bbi$. 
We make the assumption generally made that both the random effect $\bbi$ and the error term $\epsi$ 
follow a Gaussian distribution~\citep{Pinheiro2000}. 
The matrices $\Zi$ assigning the random effects to group-level observations are fixed. 
The linear covariables $\Xi$ and the nonparametric and potentially high-dimensional covariables $\Wi$ are observed and random, and they may be dependent. Furthermore, the nonparametric covariables may contain nonlinear transformations and interaction terms of the linear ones.
Please see Assumption~\ref{assumpt:Distribution} in Section~\ref{sect:ModelAndDML} for further details. 

Our aim is to estimate and make inference for the so-called fixed effect $\betazero$ in~\eqref{eq:initPLMM} in the presence of a highly complex $g$ using general machine learning 
algorithms. 
The parametric component $\betazero$ 
provides a simple summary of the covariate effects that are of main scientific interest. The nonparametric component $\g$ enhances model flexibility  because time trends and further covariates with possibly nonlinear and interaction effects can be modeled nonparametrically. 
\\

Repeated measurements, or longitudinal, data is omnipresent
in empirical research.
For example, assume we want to study the effect of a treatment over time. Observing the same subjects repeatedly presents three main advantages over having cross-sectional data. 
First, subjects can serve as their own controls.  
Second, the between-subject variability is explicitly modeled and can be excluded from the experimental error. This yields more efficient estimators of the relevant model parameters. 
Third, data can be collected more reliably~\citep{Davies2002, Fitzmaurice2011}.
\\

Various approaches have been considered in the literature to estimate the nonparametric component 
$g$ in~\eqref{eq:initPLMM}: 
kernel methods~\citep{Hart-Wehrly1986, Zeger-Diggle1994, Taavoni2019b, Chen2017},
backfitting~\citep{Zeger-Diggle1994, Taavoni2019b}, 
spline methods~\citep{Rice-Silverman1991, Zhang2004, Qin-Zhu2007, Qin-Zhu2009, Li-Zhu2010, Kim2017, Aniley2019}, 
and local linear regression~\citep{Taavoni2019b, Liang2009}. 

Our aim is to make inference for $\betazero$ in the presence of potentially highly complex effects of $\Wi$ on $\Xi$ and $\Yi$. First, we adjust $\Xi$ and $\Yi$ for $\Wi$ by regressing $\Wi$ out of them using  machine learning algorithms. These machine learning algorithms may yield biased results, especially if regularization methods are used, like for instance with the lasso~\citep{Tibshirani1996}. 
Second, we fit a linear mixed-effects model to these regression residuals to estimate $\betazero$. Our estimator of $\betazero$ converges at the parametric rate, follows a Gaussian distribution asymptotically, and is semiparametrically efficient.

We rely on the double machine learning 
framework of~\citet{Chernozhukov2018} to estimate $\betazero$ using general machine learning algorithms. 
To the best of our knowledge, this is the first work to allow the nonparametric nuisance components of a 
partially linear mixed-effects model to be estimated with arbitrary machine learners like random forests~\citep{Breiman2001} or the lasso~\citep{Tibshirani1996, Buehlmann2011}. 
In contrast to the setting and proofs of~\citet{Chernozhukov2018}, we have dependent data and need to incorporate this accordingly. 

\citet{Chernozhukov2018} introduce double machine learning to estimate a low-dimensional parameter in the presence of nonparametric nuisance components using machine learning 
algorithms. 
This estimator converges at the parametric rate and is asymptotically Gaussian due to Neyman orthogonality and sample splitting with cross-fitting. 
We would like to remark that nonparametric components can be estimated without sample splitting and cross-fitting if the underlying function class satisfies some entropy conditions; 
see for instance~\citet{Geer-Mammen1997}. However, these regularity conditions limit the complexity of the function class, and 
machine learning algorithms usually do not satisfy them. Particularly, these conditions fail to hold if the dimension of the nonparametric variables increases with the sample size~\citep{Chernozhukov2018}.

\subsection{Additional Literature}

Expositions and overviews of mixed-effects modeling techniques can be found in~\citet{Pinheiro1994, Davidian-Giltinan1995, Vonesh-Chinchilli1997, Pinheiro2000, Davidian-Giltinan2003}.
\\

\citet{Zhang1998} consider partially linear
mixed-effects models and estimate the nonparametric component with natural cubic splines. They treat the smoothing parameter as an extra variance component that is jointly estimated with the other variance components of the model. 
\citet{Masci2019} consider partially linear
mixed-effects models for unsupervised classification with discrete random effects. 
\citet{Schelldorfer2011} consider  high-dimensional linear mixed-effects models where the number of fixed effects coefficients may be much larger than the overall sample size. 
\citet{Taavoni2019a} employ a regularization approach in generalized 
partially linear mixed-effects models using regression splines to approximate the nonparametric component. 
\citet{gamm4} use penalized regression splines where the penalized components are treated as random effects.
\\

The unobserved random variables in the partially linear mixed-effects model~\eqref{eq:initPLMM} 
are assumed to follow a Gaussian distribution. 
\citet{Taavoni2021} introduce multivariate $t$ 
partially linear
mixed-effects models for longitudinal data. They consider $t$-distributed random effects to account for outliers in the data. 
\citet[Chapter 4]{Fahrmeir2011} relax the assumption of Gaussian random effects in generalized linear mixed models. They consider nonparametric Dirichlet processes and Dirichlet process mixture priors for the random effects.
\citet[Chapter 3]{Ohinata2012} consider partially linear mixed-effects models and make no distributional assumptions for the random terms, 
and the nonparametric component is estimated with kernel methods.
\citet{Lu2016} consider a partially linear mixed-effects model that is nonparametric in time and that features asymmetrically distributed errors  and missing data.

Furthermore, methods have been developed to analyze repeated measurements data that are robust to outliers. 
\citet{Guoyou2008} consider robust estimating equations and estimate the nonparametric component with a regression spline.
\citet{Tang2015} consider median-based regression methods in a partially linear model with longitudinal data to account for highly skewed responses. 
\citet{Lin2018} present an estimation technique in partially linear models for longitudinal data that is doubly robust in the sense that  it simultaneously accounts for missing responses and mismeasured covariates. 
\\

It is prespecified in 
the  partially linear mixed-effects model~\eqref{eq:initPLMM} which covariates are modeled with random effects. 
Simultaneous variable selection for fixed effects variables and random effects has been developed by~\citet{Bondell2010, Ibrahim2011}. They use penalized likelihood approaches.  
\citet{Li-Zhu2010} use a nonparametric test to test the existence of random effects in partially linear mixed-effects models. 
\citet{Zhang2020} propose a variable selection procedure for the linear covariates of a generalized partially linear model with longitudinal data. 
\\

\textit{Outline of the Paper.}
Section~\ref{sect:ModelAndDML} presents our double machine learning estimator of the linear coefficient in a partially linear mixed-effects model. Section~\ref{sect:numerical-experiments} presents our numerical results. 
\\

\textit{Notation.}
We denote by $\indset{\NN}$ the set $\{1,2,\ldots,\NN\}$. We add the probability law as a subscript to the probability operator $\Prob$ and the expectation operator $\E$ whenever we want to emphasize the corresponding dependence.
We denote the $L^p(\PP)$ norm by $\normP{\cdot}{p}$ and the Euclidean or  operator
norm  by $\norm{\cdot}$, depending on the context. 
We implicitly assume that given expectations and conditional expectations exist. We denote by $\stackrel{d}{\rightarrow}$ convergence in distribution.
The symbol $\independent$ denotes independence of random variables.
We denote by $\one_{n}$ the $n\times n$ identity matrix and omit the subscript $n$ if we do not want to emphasize the dimension. 
We denote the $d$-variate Gaussian distribution by $\mathcal{N}_d$.

\section{Model Formulation and the Double Machine Learning Estimator}\label{sect:ModelAndDML}

We consider repeated measurements data that is grouped according to experimental units or subjects. 
This grouping structure introduces dependency in the data. The individual experimental units or groups are assumed to be independent, but there may be some 
between-group heterogeneity and within-group correlation.
We consider the partially linear mixed-effects model 
\begin{equation}\label{eq:PLMM}
	\Yi = \Xi\betazero + \g(\Wi) + \Zi\bbi+ \epsi, \quad i\in\indset{\NN}
\end{equation}
for groups $i$ as in~\eqref{eq:initPLMM} to model the 
between-group heterogeneity and within-group correlation with random effects.
We have $\nni$ observations per group that are concatenated row-wise into $\Yi\in\R^{\nni}$, $\Xi\in\R^{\nni\times d}$, and $\Wi\in\R^{\nni\times \vvv}$. 
The nonparametric  variables may be high-dimensional, but $d$ is fixed.
Both $\Xi$ and $\Wi$ are random. 
The  $\Xi$ and  $\Wi$ belonging to the same group $i$ may be dependent.
For groups $i\neq j$, we assume $\Xi\independent\Xj$, $\Wi\independent\Wj$, and $\Xi\independent\Wj$.
We assume that $\Zi\in\R^{\nni\times q}$ is fixed. 
The random variable $\bbi\in\R^q$ denotes a group-specific vector of random regression coefficients that is assumed to follow a Gaussian distribution. 
The dimension $q$ of the random effects model is fixed.
Also the error terms are assumed to follow a Gaussian distribution as is commonly used 
in a mixed-effects models framework~\citep{Pinheiro2000}. 
All groups $i$ share the common linear coefficient $\betazero$ and the potentially complex function $\g\colon\R^{\vvv}\rightarrow\R$. The function $g$ is applied row-wise to $\Wi$, denoted by $\g(\Wi)$.

We denote the total number of observations by $\NNT:=\sum_{i=1}^{\NN}\nni$. 
We assume that the numbers $\nni$ of within-group observations are uniformly upper bounded by $\nmax<\infty$. 
Asymptotically, the number of groups, $\NN$,  goes to infinity. 

Our distributional and independency assumptions are summarized as follows:

\begin{assumptions}\label{assumpt:Distribution} Consider the partially linear mixed-effects model~\eqref{eq:PLMM}. We assume that there is some $\sigmazero>0$ and some symmetric positive definite matrix 
$\Gammazero\in\R^{q\times q}$ 
such that the following conditions hold.
\begin{enumerate}[label={\theassumptions.\arabic*}]
	\item \label{assumpt:D1} 
		The random effects $\bone,\ldots,\bN$ are \iid\ $\mathcal{N}_q(\bo, \Gammazero)$.
	\item \label{assumpt:D2}
		The error terms $\epsone,\ldots,\epsN$ are independent
		and follow a Gaussian distribution, $\epsi\sim\mathcal{N}_{\nni}(\bo, \sigmazero^2\oneni)$ for $i\in\indset{\NN}$, with the common variance component $\sigmazero^2$. 
	\item \label{assumpt:D3}
		The variables $\bone,\ldots,\bN, \epsone,\ldots,\epsN$ are independent. 
	\item \label{assumpt:D4-2}
		For all $i, j \in\indset{\NN}$, $i\neq j$, we have $(\bbi,\epsi)\independent(\Wi, \Xi)$ and $(\bbi,\epsi)\independent(\Wj, \Xj)$.
	\item \label{assumpt:D6}
	         For all $i, j \in\indset{\NN}$, $i\neq j$, we have $\Xi\independent\Xj$, $\Wi\independent\Wj$, and $\Xi\independent\Wj$.
\end{enumerate}
\end{assumptions}

We would like to remark that the distribution of the error terms $\epsi$ in Assumption~\ref{assumpt:D2} can be generalized to $\epsi\sim\mathcal{N}_{\nni}(\bo,\sigmazero^2\Lambda_i(\blambda))$, where $\Lambda_i(\blambda)\in\R^{\nni\times\nni}$ is a symmetric positive definite matrix parametrized by some finite-dimensional parameter vector $\blambda$ that all groups have in common. 
For the sake of notational simplicity, we restrict ourselves to Assumption~\ref{assumpt:D2}. 

Moreover, we may consider stochastic random effects matrices $\Zi$. Alternatively, the nonparametric variables $\Wi$ may be part of the random effects matrix. In this case, 
we consider 
the random effects matrix $\Ztili = \h(\Zi,\Wi)$ for some known function $\h$ in~\eqref{eq:PLMM} instead of $\Zi$.  
Please see Section~\ref{sec:nonfixedZi} in the appendix for further details. 
For simplicity, we restrict ourselves to fixed random effects matrices $\Zi$ that are disjoint from $\Wi$. 
\\

The unknown parameters in our model are $\betazero$, $\Gammazero$, and $\sigmazero$. 
Our aim is to estimate 
$\betazero$ and 
make inference for it.
Although the variance parameters $\Gammazero$ and $\sigmazero$ need to be estimated consistently to construct an estimator of $\betazero$, it is not our goal to perform inference for them.

\subsection{The Double Machine Learning Fixed-Effects Estimator}

Subsequently, we describe our estimator of $\betazero$ in~\eqref{eq:PLMM}. 
To motivate our procedure, we first consider the population version with the residual terms
\begin{displaymath}
	\RXi := \Xi - \E[\Xi | \Wi] \quad\textrm{and}\quad
	\RYi := \Yi - \E[\Yi | \Wi] \quad\textrm{for}\quad
	i\in\indset{\NN}
\end{displaymath}
that adjust $\Xi$ and $\Yi$ for $\Wi$. On this adjusted level, we have the linear mixed-effects model
\begin{equation}\label{eq:LMM}
	\RYi = \RXi \betazero + \Zi\bbi + \epsi, \quad i\in\indset{\NN}
\end{equation}
due to~\eqref{eq:PLMM} and Assumption~\ref{assumpt:D4-2}. 
In particular, the adjusted and grouped responses in this model are independent in the sense that we have $\RYi \independent \RYj$ for $i\neq j$. 
The strategy now is to first estimate the residuals with machine learning algorithms and then use 
linear mixed model techniques to infer $\beta_0$. This is done with sample splitting and cross-fitting, and the details are described next.
\\

Let us define $\Sigmazero := \sigmazero^{-2}\Gammazero$ 
and $\Vzeroi:=(\Zi\Sigmazero\Zi^T + \oneni)$ so that we have
\begin{equation}\label{eq:resNormal}
	(\RYi | \Wi,\Xi) \sim \mathcal{N}_{\nni}\big(\RXi\betazero, \sigmazero^2\Vzeroi\big). 
\end{equation}
We assume that there exist functions $\mX^0\colon\R^{\vvv}\rightarrow\R^d$ and $\mY^0\colon\R^{\vvv}\rightarrow\R$ that we can apply row-wise to $\Wi$ to have
$\E[\Xi | \Wi] = \mX^0(\Wi)$ and $\E[\Yi | \Wi] = \mY^0(\Wi)$. 
In particular, $\mX^0$ and $\mY^0$ do not depend on the grouping index $i$. 
Let $\etazero:= (\mX^0,\mY^0)$ denote the true unknown nuisance parameter. 
Let us denote by $\thetazero := (\betazero, \sigmazero^2,\Sigmazero)$ the complete true unknown parameter vector and by $\theta := (\beta,\sigma^2,\Sigma)$ and $\Vi:=\Zi\Sigma\Zi^T+\oneni$ respective general parameters. 
The log-likelihood of group $i$ is given by 
\begin{equation}\label{eq:log-likelihood}
\begin{array}{rcl}
	\libig{\theta, \etazero} &=& -\frac{\nni}{2} \log(2\pi) - \frac{\nni}{2}\log(\sigma^2) - \frac{1}{2}\log\big( \det(\Vi)\big) \\
	&&\quad- \frac{1}{2\sigma^2} (\RYi-\RXi\beta)^T\Vi^{-1}(\RYi-\RXi\beta)  - \log \big(p(\Wi,\Xi)\big),
	\end{array}
\end{equation}
where $p(\Wi,\Xi)$ denotes the joint density of $\Wi$ and $\Xi$. We assume that $p(\Wi,\Xi)$ does not depend on $\theta$. 
The true nuisance parameter $\etazero$ in the log-likelihood~\eqref{eq:log-likelihood} is unknown and estimated with machine learning algorithms (see below).
Denote by $\eta := (\mX,\mY)$ some general nuisance parameter. The terms that adjust $\Xi$ and $\Yi$ for $\Wi$ with this general nuisance parameter are given by $\Xi-\mX(\Wi)$ and $\Yi-\mY(\Wi)$. Up to additive constants that do not depend on $\theta$ and $\eta$, we thus consider maximum likelihood estimation with the likelihood
\begin{displaymath}
	\begin{array}{rl}
	&\li(\theta,\eta) 
	= -\frac{\nni}{2}\log(\sigma^2) - \frac{1}{2}\log\big(\det(\Vi) \big)\\
	&\quad- \frac{1}{2\sigma^2} \Big(\Yi-\mY(\Wi)-\big(\Xi-\mX(\Wi)\big)\beta\Big)^T\Vi^{-1}\Big(\Yi-\mY(\Wi)-\big(\Xi-\mX(\Wi)\big)\beta\Big),
	\end{array}
\end{displaymath}
which is a function of both the finite-dimensional parameter $\theta$ and the infinite-dimensional
nuisance parameter $\eta$. 
\\

Our estimator of $\betazero$ is constructed as follows using double machine learning. 
First, we estimate $\etazero$ with machine learning 
algorithms and plug these estimators into the estimating equations for $\thetazero$, equation~\eqref{eq:setToZero} below, 
to obtain an estimator for $\beta_0$. This procedure is done with sample splitting and cross-fitting as explained next.

Consider repeated measurements from $\NN$ experimental units, subjects, or groups as in~\eqref{eq:PLMM}. Denote by $\Si :=(\Wi,\Xi,\Zi,\Yi)$ the observations of group $i$.
First, we split the group indices $\indset{\NN}$ into $\KK\ge 2$ disjoint sets $I_1, \ldots, I_{\KK}$ of approximately equal size; 
please see Section~\ref{sect:AssumptionsDefinitions} in the appendix for further details. 

For each $\kk\in\indset{\KK}$, we estimate the conditional expectations $\mX^0(W)$ and $\mY^0(W)$ with data from $\Ikc$. We call the resulting estimators  $\hmX^{\Ikc}$ and $\hmY^{\Ikc}$, respectively. 
Then, the adjustments $\hbRkXi:=\Xi-\hmX^{\Ikc}(\Wi)$, and $\hbRkYi:=\Yi-\hmY^{\Ikc}(\Wi)$ for $i\in\Ik$ are evaluated on $\Ik$, the complement of $\Ikc$. 
Let $\hetaIkc := (\hmX^{\Ikc}, \hmY^{\Ikc})$ denote the estimated nuisance parameter.
Consider the score function
	$\psi(\Si; \theta,\eta) := \nabla_{\theta}\li(\theta,\eta)$, 
where $\nabla_{\theta}$ denotes the gradient with respect to $\theta$ interpreted as a vector.
On each set $\Ik$, we consider an estimator $\hthetak = (\hbetak,\hsigmak^2, \hSigmak)$ of $\thetazero$ that, approximately, in the sense of Assumption~\ref{assumpt:Theta2} in the appendix, solves
\begin{equation}\label{eq:setToZero}
	\frac{1}{\nnktot}\sum_{i\in\Ik} \psibig{\Si;\hthetak,\hetaIkc} 
	= \frac{1}{\nnktot}\sum_{i\in\Ik} \nabla_{\theta}\li(\theta,\eta)
	\stackrel{!}{=}  \bo,
\end{equation}
where $\nnktot := \sum_{i\in\Ik}\nni$ denotes the total number of observations from experimental units that belong to the set $\Ik$.
These $\KK$ estimators $\hthetak$ for $\kk\in\indset{\KK}$ are assembled to form the final cross-fitting estimator
\begin{equation}\label{eq:betahat}
	\hbeta := \frac{1}{\KK} \sum_{\kk=1}^{\KK} \hbetak
\end{equation}
of $\betazero$. 
We remark that one can simply use linear mixed model computation and software to compute $\hbetak$ based on the estimated residuals $\RIk$. 
The estimator $\hbeta$ fundamentally depends on the particular sample split. To alleviate this effect, the overall procedure may be repeated $\Salg$ times~\citep{Chernozhukov2018}. The $\Salg$ point estimators are aggregated by the median, and an additional term accounting for the random splits is added to the variance estimator of $\hbeta$; please see 
Algorithm~\ref{algo:Summary} that presents the complete procedure.

\begin{algorithm}[h!]
	\SetKwInOut{Input}{Input}
   	\SetKwInOut{Output}{Output}

 \Input{$\NN$ \iid\  observations $\{\Si=(\Wi,\Xi, \Zi, \Yi)\}_{i\in\indset{\NN}}$ from the model~\eqref{eq:PLMM} satisfying Assumption~\ref{assumpt:Distribution}, a natural number $\KK$, a natural number $\Salg$.}
 \Output{An estimator of $\betazero$ in~\eqref{eq:PLMM} together with its estimated asymptotic variance.}
 
 \For{$s\in\indset{\Salg}$}
 {
 Split the grouped observation index set $\indset{\NN}$ into $\KK$ sets $I_1, \ldots, I_{\KK}$ of approximately equal size.
 
 \For{$\kk\in\KK$}
 {
 Compute the conditional expectation estimators $\hmX^{\Ikc}$ and $\hmY^{\Ikc}$ with some  
 machine learning algorithm and data from $\Ikc$. 
 
 Evaluate the adjustments $\hbRkXi=\Xi-\hmX^{\Ikc}(\Wi)$ and $\hbRkYi=\Yi-\hmY^{\Ikc}(\Wi)$ for $i\in\Ik$. 
 
Compute $\hthetaks = (\hbetaks, \hsigmaks^2, \hSigmaks)$ using, for instance, linear mixed model techniques. 
 }

Compute $\hbetas = \frac{1}{\KK}\sum_{\kk=1}^{\KK}\hbetaks$ as an approximate solution to~\eqref{eq:setToZero}.

Compute an estimate $\hTzeros$ of the asymptotic variance-covariance matrix $\Tzero$ in Theorem~\ref{thm:asymptoticGaussian}. 
 }
 
 Compute $\hbeta = \mathrm{median}_{s\in\indset{\Salg}}(\hbetas)$.
 
 Estimate $\Tzero$ by $\hTzero = \mathrm{median}_{s\in\indset{\Salg}}(\hTzeros + (\hbeta-\hbetas)(\hbeta-\hbetas)^T)$.\label{algo:varianceCorrection}

 \caption{Double machine learning in a partially linear mixed-effects model with repeated measurements.}\label{algo:Summary}
\end{algorithm}

\subsection{Theoretical Properties of the Fixed-Effects Estimator}

The estimator $\hbeta$ as in~\eqref{eq:betahat} converges at the parametric rate, $\NN^{-1/2}$, and is asymptotically Gaussian distributed and semiparametrically efficient. 

\begin{theorem}\label{thm:asymptoticGaussian}
Consider grouped observations $\{\Si=(\Wi,\Xi,\Yi)\}_{i\in\indset{\NN}}$ from the partially linear mixed-effects model~\eqref{eq:PLMM} that satisfy Assumption~\ref{assumpt:Distribution} such that $p(\Wi,\Xi)$ does not depend on $\theta$. 
Let $\NNT:=\sum_{i=1}^{\NN}\nni$ denote the total number of unit-level observations. 
Furthermore, suppose the assumptions in Section~\ref{sect:AssumptionsDefinitions} in the appendix hold,
and consider the symmetric positive-definite matrix $\Tzero$ given in Assumption~\ref{assumpt:regularity7} in the appendix. 
Then, $\hbeta$ as in~\eqref{eq:betahat} concentrates in a $1/\surd{\NNtot}$ neighborhood of $\betazero$ and is centered Gaussian, namely
\begin{equation}\label{eq:thmEquation}
	\surd{\NNtot}\Tzero^{\frac{1}{2}}(\hbeta-\betazero)
	\stackrel{d}{\rightarrow}\mathcal{N}_d(\bo, \one_{d}) \quad (\NN\rightarrow\infty), 
\end{equation}
and semiparametrically efficient. 
The convergence in~\eqref{eq:thmEquation} is in fact uniformly over the law $P$ of $\{\Si=(\Wi,\Xi,\Yi)\}_{i\in\indset{\NN}}$. 
\end{theorem}

Please see Section~\ref{sect:asymptoticDistribution} in the appendix for a proof of Theorem~\ref{thm:asymptoticGaussian}. 
Our proof builds on~\citet{Chernozhukov2018}, but we have to take into account the correlation within units that is introduced by the random effects.

The inverse asymptotic variance-covariance matrix $\Tzero$ can be consistently estimated;
see Lemma~\ref{lem:multiplierMatrix} in the appendix. 
The estimator $\hbeta$ is semiparametrically efficient because the score function comes from the log-likelihood of our data and because $\etazero$ solves a concentrating-out equation for fixed $\theta$; see~\citet{Chernozhukov2018, Newey1994b}. 

The assumptions in Section~\ref{sect:AssumptionsDefinitions} of the appendix specify regularity conditions and  required 
convergence rates of the machine learning estimators. The 
machine learning errors need to satisfy the product relationship 
\begin{displaymath}
	\normP{\mX^0(W)-\hmX^{\Ikc}(W)}{2}\big(\normP{\mY^0(W)-\hmY^{\Ikc}(W)}{2}+ \normP{\mX^0(W)-\hmX^{\Ikc}(W)}{2} \big)\ll\NN^{-\frac{1}{2}}.
\end{displaymath}
This bound requires that only the products of the  
machine learning estimation errors 
$\normP{\mX^0(W)-\hmX^{\Ikc}(W)}{2}$ and $\normP{\mY^0(W)-\hmY^{\Ikc}(W)}{2}$
but not the individual ones need to vanish at a rate smaller than
$\NN^{-1/2}$. In particular, the individual estimation errors  may vanish at the  rate 
smaller than $\NN^{-1/4}$. This is achieved by many machine learning  
methods (cf. 
\citet{Chernozhukov2018}):
$\ell_1$-penalized and related methods in a variety of sparse models
\citep{Bickel2009, Buehlmann2011, Belloni2011, Belloni-Chernozhukov2011, Belloni2012, Belloni-Chernozhukov2013}, forward selection in sparse models
\citep{Kozbur2020}, $L_2$-boosting in sparse linear models
\citep{Luo2016}, a class of regression trees and random forests 
\citep{Wager2016}, and neural networks \citep{Chen1999}. 

We note that so-called Neyman orthogonality makes score functions insensitive to inserting potentially biased machine learning 
estimators of the nuisance parameters. A score function is Neyman orthogonal if its Gateaux derivative vanishes at the true $\thetazero$ and the true $\etazero$. In particular, Neyman orthogonality is a first-order property. The product relationship of the machine learning 
estimating errors described above is used to bound second-order terms. We refer to Section~\ref{sect:asymptoticDistribution} in the appendix for more details.

\section{Numerical Experiments}\label{sect:numerical-experiments}

We apply our method to an empirical and a pseudorandom dataset and in a simulation study. 
Our implementation is available in the \textsf{R}-package \texttt{dmlalg}~\citep{dmlalg}.

\subsection{Empirical Analysis: CD4 Cell Count Data}

Subsequently, we apply our method to longitudinal CD4 cell counts data collected from human immunodeficiency virus (HIV) seroconverters. This data has previously been analyzed by~\citet{Zeger-Diggle1994} and is available in the \textsf{R}-package \texttt{jmcm}~\citep{jmcm} as \texttt{aids}. 
It contains $2376$ observations of CD4 cell counts measured on $369$ subjects. The data was collected during a period ranging from $3$ years before to $6$ years after seroconversion. The number of observations per subject ranges from $1$ to $12$, but for most subjects,  $4$ to $10$ observations are available. 
Please see~\citet{Zeger-Diggle1994} for more details on this dataset. 

Apart from time, five other covariates are measured: the age at seroconversion in years (age), the smoking status measured by the number of cigarette packs consumed per day (smoking), a binary variable indicating drug use (drugs), the number of sex partners (sex), and the depression status measured on the Center for Epidemiologic Studies Depression (CESD) scale (cesd), where higher CESD values indicate the presence of more depression symptoms. 

We incorporate a random intercept per person. 
Furthermore, we consider a square-root transformation of the CD4 cell counts to reduce the skewness of this variable as proposed by~\citet{Zeger-Diggle1994}. 
The CD4 counts are our response. The covariates that are of scientific interest are considered as $X$'s, 
and the remaining covariates are considered as $W$'s in the partially linear mixed-effects model~\eqref{eq:PLMM}.
The effect of time is modeled nonparametrically, but there are several options to model the other covariates. 
Other models than partially linear mixed-effects model
have also been considered in the literature to analyze this dataset. 
For instance, \citet{Fan2000} consider a functional linear model where the linear coefficients are a function of the time. 
\\

We consider two partially linear mixed-effects models for this dataset.
First, we incorporate all covariates except time linearly. 
Most approaches in the literature considering a partially linear mixed-effects model for this data that model time nonparametrically
report that sex and cesd are significant and that either smoking or drugs is significant as well; see for instance~\citet{Zeger-Diggle1994, Taavoni2019b, Wang2011}. 
\citet{Guoyou2008} develop a robust estimation method 
for longitudinal data and estimate 
nonlinear effects from time with regression splines. 
With the CD4 dataset, 
They find that smoking and cesd are significant. 

We apply our method with $\KK=2$ sample splits, $\mathcal{S}=100$ repetitions of splitting the data, and learn the conditional expectations with random forests that consist of $500$ trees whose minimal node size is $5$. 

Like~\citet{Guoyou2008},
we conclude that smoking and cesd are significant; please see the first row of Table~\ref{tab:empirical} for a more precise account of our findings.
Therefore,
we can expect
that our method implicitly performs robust estimation. Apart from sex, our point estimators are larger or of about the same size in absolute value as what~\citet{Guoyou2008} obtain. This suggests that our method incorporates potentially less bias. However, apart from age, the standard deviations are slightly larger with our method. 
This can be expected because random forests are more complex than the regression splines~\citet{Guoyou2008} employ. 
\\

\begingroup
\begin{table}
\footnotesize
\centering
\begin{tabular}{| L{2.5cm} | C{1.9cm}  | C{1.9cm} | C{1.9cm} | C{1.9cm} | C{2.2cm} |}
\hline
&  \textbf{age} & \textbf{smoking} & \textbf{drugs} & \textbf{sex} & \textbf{cesd}\\
\hline
\hline
$W = (\mathrm{time})$ &  $0.004$ ($0.027$) & $0.752$ ($0.123$) & $0.704$ ($0.360$) & $0.001$ ($0.043$)    &   $-0.042$ ($0.015$)\\
\hline
$W = (\mathrm{time}, \mathrm{age}, \mathrm{sex})$ & - & $0.620$ ($0.126$) & $0.602$ ($0.335$) &- & $-0.047$ ($0.015$) \\
\hline
\hline
\citet{Zeger-Diggle1994} 
& $0.037$ ($0.18$) & $0.27$ ($0.15$) & $0.37$ ($0.31$) & $0.10$ ($0.038$) &  $-0.058$ ($0.015$) \\
\hline
\citet{Taavoni2019b} 
& $1.5\cdot 10^{-17}$ ($3.5\cdot 10^{-17}$)& $0.152$ ($0.208$) & $0.130$ ($0.071$) &  $0.0184$ ($0.0039$) & $-0.0141$ ($0.0061$)  \\
\hline
\citet{Wang2011} 
& $0.010$ ($0.033$) & $0.549$ ($0.144$) & $0.584$ ($0.331$) & $0.080$ ($0.038$) & $-0.045$ ($0.013$)\\
\hline
\citet{Guoyou2008} 
& $0.006$ ($0.038$) &  $0.538$ ($0.136$) & $0.637$ ($0.350$)& $0.066$ ($0.040$) & $-0.042$ ($0.015$)\\
\hline
\end{tabular}
\caption{\label{tab:empirical}Estimates of the linear coefficient and its standard deviation in parentheses with our method 
for nonparametrically adjusting for time (first row) and for time, age, and sex (second row).
The remaining rows display the results from~\citet[Section 5]{Zeger-Diggle1994}, \citet[Table 1, ``Kernel'']{Taavoni2019b}, \citet[Table 2, ``Semiparametric efficient scenario I'']{Wang2011}, and \citet[Table 5, ``Robust'']{Guoyou2008}, respectively. 
}
\end{table}
\endgroup

We consider a second estimation approach where we model the variables time, age, and sex nonparametrically and allow them to interact. 
It is conceivable that these variables are not (causally) influenced by smoking, drugs, and cesd and that they are therefore exogenous. 
The variables smoking, drugs, and cesd are modeled linearly,  
and they are considered as treatment variables. Some direct causal effect interpretations are possible if one is willing to assume, for instance, that the nonparametric adjustment variables are causal parents of the linear variables or the response. 
However, we do not pursue this line of thought further.

We estimate the conditional expectations given the three nonparametric variables time, age, and sex again with random forests that consist of $500$ trees whose minimal node size is $5$ and use $K=2$ and $\Salg=100$ in Algorithm~\ref{algo:Summary}. 
We again find that smoking and cesd are significant; please see the second row of Table~\ref{tab:empirical}. 
This cannot be expected a priori because this second model incorporates more complex adjustments, which can lead to less significant variables.

\subsection{Pseudorandom Simulation Study: CD4 Cell Count Data}\label{sect:pseudorandom}

Second, we consider the CD4 cell count data from the previous subsection and perform a pseudorandom simulation study. The variables smoking, drugs, and cesd are modeled linearly and the variables time, age, and sex nonparametrically. 
We condition on these six variables in our simulation. That is, they are the same in all repetitions. 
The function $\g$ in~\eqref{eq:PLMM} is chosen as a regression tree that we built beforehand. 
We let $\betazero=(0.62, 0.6, -0.05)^T$, where the first component corresponds to smoking, the second one to drugs, and the last one to cesd, consider a  standard deviation of the random intercept per subject of $4.36$, and a standard deviation of the error term of $4.35$. These are the point estimates of the respective quantities obtained in the previous subsection. 

Our fitting procedure uses random forests consisting of $500$ trees whose minimal node size is $5$ to estimate the conditional expectations, and we use $K=2$ and $\Salg=10$ in Algorithm~\ref{algo:Summary}. We perform $5000$ simulation runs.
We compare the performance of our method with that of the spline-based function \texttt{gamm4} from the package~\texttt{gamm4}~\citep{gamm4} for the statistical software \textsf{R}~\citep{R}.
This method represents the nonlinear part of the model by smooth additive functions and estimates them by penalized regression splines. The penalized components are treated as random effects and the unpenalized components 
as fixed. 

The results are displayed in Figure~\ref{fig:pseudorandom}. 
With our method,  \texttt{mmdml}, the two-sided confidence intervals for $\betazero$ are of about the same length  but achieve a coverage that is closer to the nominal $95\%$ level than with \texttt{gamm4}. The \texttt{gamm4} method largely undercovers the packs component of $\betazero$, which can be explained by the incorporated bias.

\begin{figure}[]
	\centering
	\caption[]{\label{fig:pseudorandom} 
	Coverage and  length of two-sided confidence intervals at significance level 5\% and  bias for our method, \texttt{mmdml}, and 
	\texttt{gamm4}. 
	In the coverage plot, solid dots represent point estimators, and circles represent $95\%$ confidence bands with respect to the $5000$ simulation runs.
	The confidence interval length and bias are displayed with boxplots without outliers.}
	\includegraphics[width=0.75\textwidth]{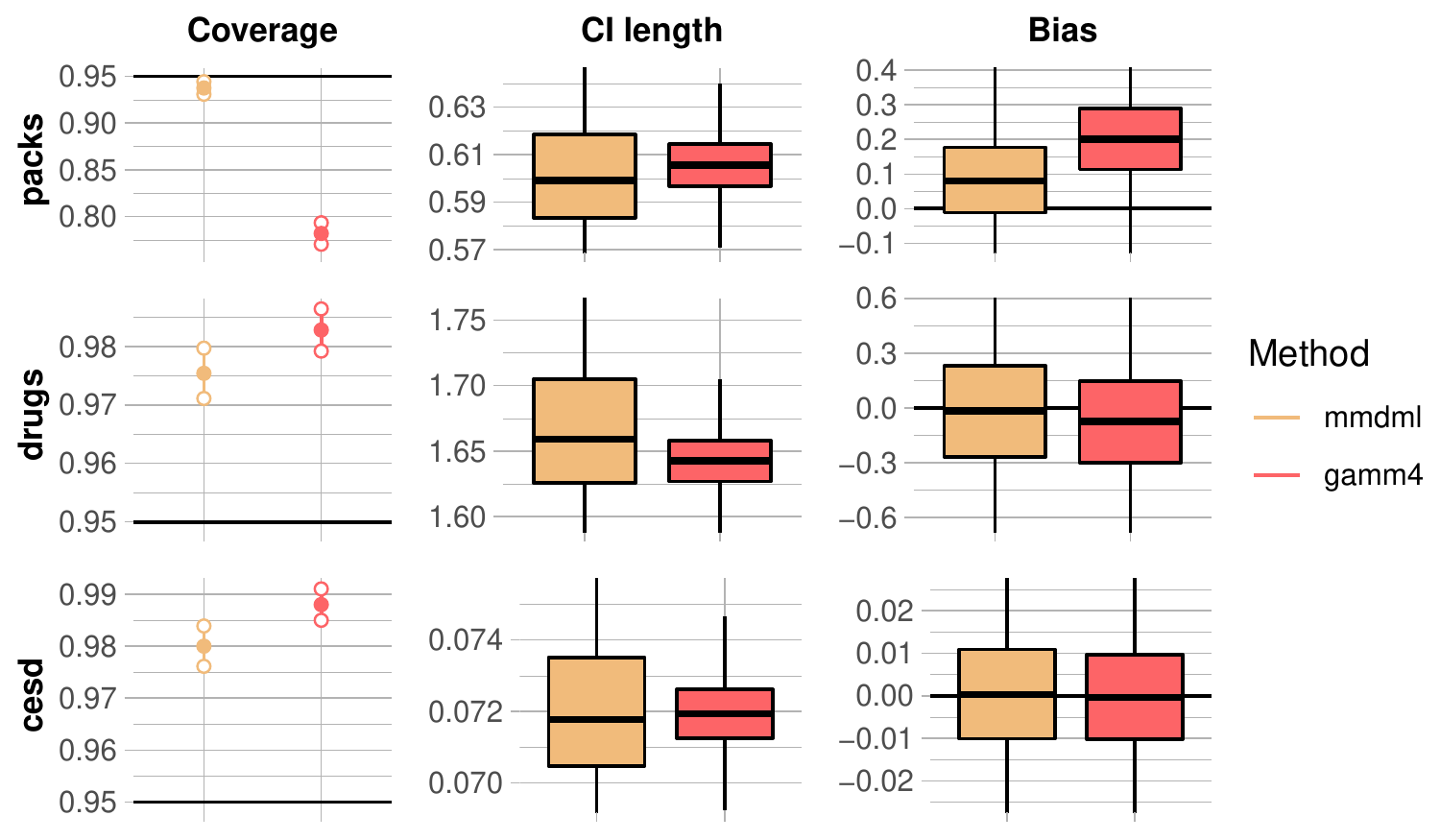}
\end{figure}

\subsection{Simulation Study}\label{sect:simul-study}

We consider a partially linear mixed-effects model with $q=3$ 
random effects and where $\betazero$ is one-dimensional.
Every subject has their own random intercept term and a nested random effect with two levels.
Thus, the random effects structure is more complex than in the previous two subsections because these models only used a random intercept. 
We compare three data generating mechanisms: 
One where the function $\g$ is nonsmooth and the number of observations per group is balanced, 
one where the function $\g$ is smooth and the number of observations per group is balanced, 
and one where the function $\g$ is nonsmooth and the number of observations per group is unbalanced;
please see Section~\ref{sect:dataSimulation} in the appendix for more details.

We estimate the nonparametric nuisance components, that is, the conditional expectations, with
random forests consisting of $500$ trees whose minimal node size is $5$. 
Furthermore, we use $K=2$ and $\Salg=10$ in Algorithm~\ref{algo:Summary}. 

We perform $1000$ simulation runs and consider different numbers of groups $N$. As in the previous subsection, we compare the performance of our method with  
\texttt{gamm4}.

The results are displayed in Figure~\ref{fig:simulation}. Our method, \texttt{mmdml}, highly outperforms \texttt{gamm4} in terms of coverage for nonsmooth $g$ because the coverage of \texttt{gamm4} equals $0$ due to its substantial bias. Our method overcovers slightly due to the correction factor that results from the $\Salg$ repetitions. However, this correction factor is highly recommended in practice.  
With smooth $g$, \texttt{gamm4} is closer to the nominal coverage and has shorter confidence intervals than our method. Because the underlying model is smooth and additive, a spline-based estimator is better suited. In all scenarios,  our method outputs longer confidence intervals than \texttt{gamm4} because we use random forests; consistent with theory, the difference in absolute value decreases though when $N$ increases.

\begin{figure}[]
	\centering
	\caption[]{\label{fig:simulation} 
	Coverage and median length of two-sided confidence intervals for $\betazero$ at significance level 5\% (true $\betazero = 0.5$) and median bias
	for three data generating scenarios for our method, \texttt{mmdml}, and 
	\texttt{gamm4}. 
	The shaded regions in the coverage plot represent $95\%$ confidence bands with respect to the $1000$ simulation runs. 
	The dots in the coverage and bias plot are jittered, but neither are their interconnecting lines nor their confidence bands.}
	\includegraphics[width=\textwidth]{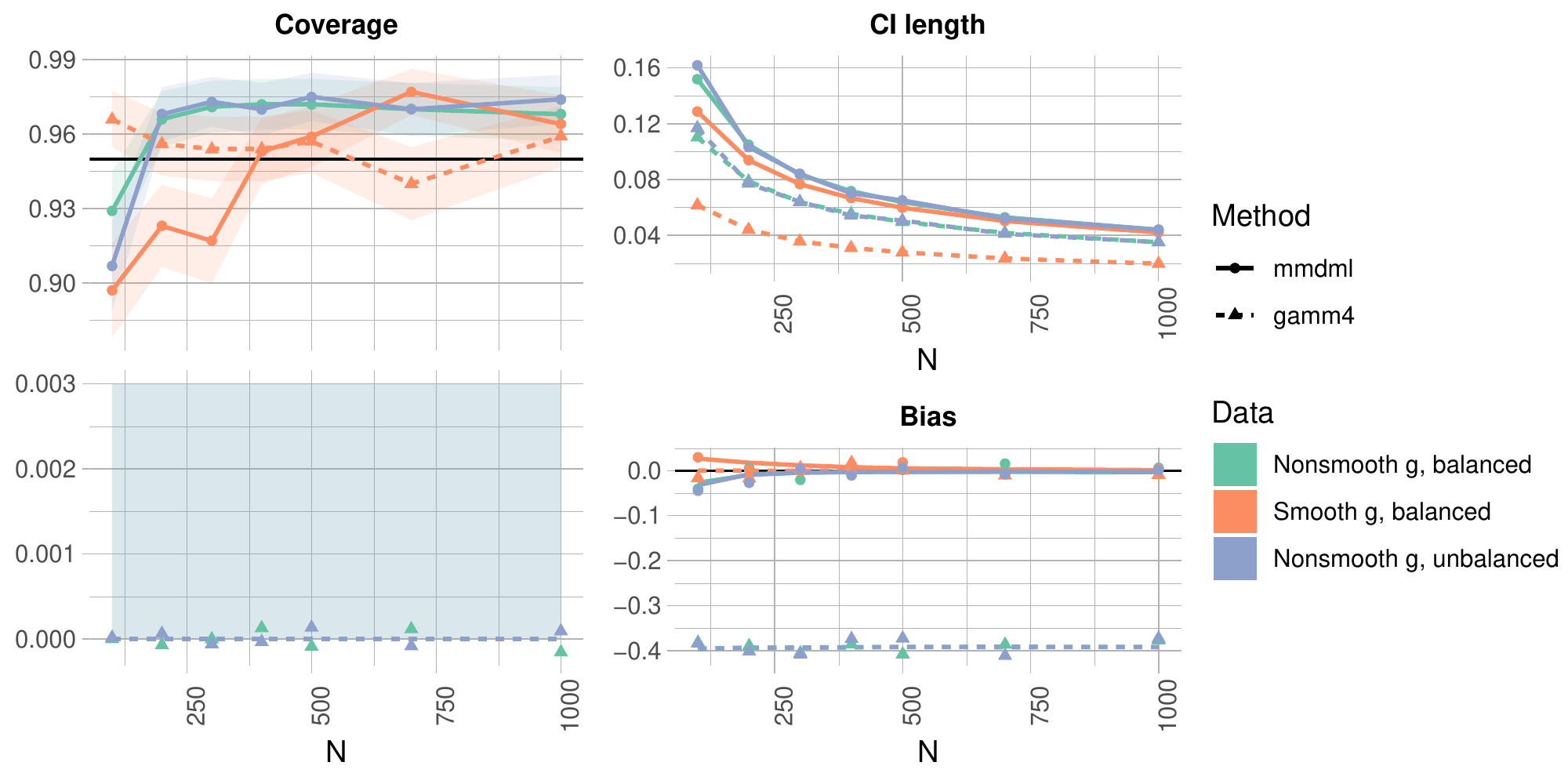}
\end{figure}

\section{Conclusion}\label{sect:conclusion}

Our aim was to develop inference for the linear coefficient $\betazero$ of a partially linear mixed-effects model 
that includes a linear term and potentially complex nonparametric terms.
Such models can be used to describe heterogenous and correlated data that feature 
some grouping structure, which may result from taking repeated measurements. 
Traditionally, spline or kernel approaches are used to cope with the nonparametric part of such a model. We presented a scheme that uses the double machine learning  
framework of~\citet{Chernozhukov2018} to estimate any 
nonparametric components with arbitrary machine learning algorithms. This allowed us to consider complex nonparametric components with interaction structures and high-dimensional variables. 

Our proposed method is as follows. First, the nonparametric variables are regressed out from the response and the linear variables. This step adjusts the response and the linear variables  for the nonparametric variables and may be performed with any machine learning algorithm.
The adjusted variables satisfy a linear mixed-effects model, where the linear coefficient $\betazero$ can be estimated with standard linear 
mixed-effects
techniques. We showed that the estimator of $\betazero$ 
asymptotically follows a Gaussian distribution, converges at the parametric rate, and is semiparametrically efficient.
This asymptotic result 
allows us to perform inference for $\betazero$. 

Empirical experiments demonstrated the performance of our proposed method.  
We conducted an empirical and pseudorandom data analysis and a  simulation study. 
The simulation study and the pseudorandom experiment confirmed the effectiveness of our method in terms of coverage, length of confidence intervals, and estimation bias
compared to a penalized regression spline approach relying on additive models. 
In the empirical experiment, we analyzed longitudinal CD4 cell counts data collected from 
HIV-infected individuals. 
In the literature, most methods only incorporate the time component nonparametrically to analyze this dataset.
Because we estimate nonparametric components with machine learning algorithms, we can allow several variables to enter the model nonlinearly, and we can allow these variables to interact. 
A comparison of our results with the literature suggests that our method  
may perform robust estimation. 

Implementations of our method are available in the \textsf{R}-package \texttt{dmlalg}~\citep{dmlalg}.

\section*{Acknowledgements}

This project has received funding from the European Research Council (ERC) under the European Union’s Horizon 2020 research and innovation programme (grant agreement No. 786461).

\phantomsection
\addcontentsline{toc}{section}{References}
\bibliography{references}

\begin{appendices}

\section{Data Generating Mechanism for Simulation Study}\label{sect:dataSimulation}

Let $n=15$. 
For all scenarios except the unbalanced one, 
we sample the number of observations for each experimental unit from $\{n-3, n-2, \ldots, n+2, n+3\}$ with equal probability. 
For the unbalanced scenario, we sample the number of observations for each experimental unit from $\{1, 2, \ldots, 2n-2, 2n-1\}$ with equal probability. 
We consider $3$-dimensional nonparametric variables. For $w=(w_1,w_2,w_{3})\in\R^{3}$, consider the real-valued functions
\begin{displaymath}
\begin{array}{cl}
	&h(w) \\
	:=& -3 \cdot \one_{w_3>0}  \one_{w_1>0}  +
    2 \cdot\one_{w_3>0} \one_{w_1\le 0}  -
    \one_{w_3\le 0}  \one_{w_3\le -1}  -
    2 \cdot\one_{w_3\le 0}  \one_{w_3>-1}  \one_{w_2>0}  \\
    &\quad -
    3 \cdot\one_{w_3\le 0}  \one_{w_3>-1}  \one_{w_2\le 0}  \one_{w_1>0.75} +
    \one_{w_3\le 0}  \one_{w_3>-1} \one_{w_2\le 0}  \one_{w_1\le 0.75} 
    \end{array}
\end{displaymath}
and 
\begin{displaymath}
\begin{array}{cl}
	&\g(w)\\ 
	:=& \one_{w_1>0}  \one_{w_2>0}  \one_{w_3>1}  -
    1.5\cdot \one_{w_1>0} \one_{w_2>0}  \one_{w_3\le 1}  -
    2.7\cdot \one_{w_1>0}  \one_{w_2\le 0}  \one_{w_2\le -0.5} \one_{w_1>1}  \one_{w_3>1.25}  \\
    &\quad - 0.5\cdot \one_{w_1>0}  \one_{w_2\le }  \one_{w_2\le -0.5}  \one_{w_1>1} \one_{w_3\le 1.25} +
    3.2 \cdot\one_{w_1>0}  \one_{w_2\le 0}  \one_{w_2\le -0.5}  \one_{w_1\le 1}  \\
    &\quad + 0.75\cdot \one_{w_1>0} \one_{w_2\le 0} \one_{w_2>-0.5}  +
    3 \cdot\one_{w_1\le 0} \one_{w_3>0}  \one_{w_2\le -1}  \one_{w_1\le -1.3} \\
    &\quad+
    1.5\cdot \one_{w_1\le 0}  \one_{w_3>0} \one_{w_2\le -1}  \one_{w_1>-1.3}  -
    2.3\cdot \one_{w_1\le 0}  \one_{w_3> 0}  \one_{w_2>-1}  \\
    &\quad+
    2.8\cdot \one_{w_1\le 0} \one_{w_3\le 0}  \one_{w_3\le -0.75}  +
    2 \cdot\one_{w_1\le 0}  \one_{w_3\le 0}  \one_{w_3>-0.75}  \one_{w_1\le -0.5}  \\
    &\quad-
    1.75 \cdot\one_{w_1\le 0} \one_{w_3\le 0}  \one_{w_3>-0.75} \one_{w_1>-0.5} 
    \end{array}
\end{displaymath}

For the nonparametric covariable, we consider the following data generating mechanism. 
The matrix $\Wi\in\R^{\nni\times 3}$ contains the $\nni$ observations of the $i$th experimental unit in its rows. We draw these $\nni$ rows of $\Wi$ independently. That is, 
$(\Wi)_{k, \cdot}\sim\mathcal{N}_{3}(\bo,\one)$ for $i\in\indset{\NN}$ and $k\in\indset{\nni}$ with $(\Wi)_{k, \cdot}\independent (\Wi)_{l, \cdot}$, $k\neq l$,$k, l\in\indset{\nni}$ and $\Wi\independent\Wj$, $i\neq j$, $i, j\in\indset{\NN}$.

The linear covariable $\Xi$ is modeled with $\Xi = h(\Wi) + \epsXi$, where its error term $\epsXi\sim\mathcal{N}_{\nni}(\bo, \one)$ for $i\in\indset{\NN}$ and $\epsXi\independent \epsXj$ for $i\neq j$, $i,j\in\indset{\NN}$. 

For $\betazero=0.5$ and $\sigma_0=1$, the model of the response $\Yi$ is 
$\Yi = \Xi\betazero + \g(\Wi) + \Zi\bbi + \epsi$ with
\begin{displaymath}
	\Zi = \begin{pmatrix} 1 & 0 & 1\\
	                                 1 & 0 & 1 \\
	                                 \vdots & \vdots & \vdots\\
	                                 1 & 0 & 1\\
	                                 0 & 1 & 1\\
	                                 0 & 1 & 1\\
	                                 \vdots & \vdots & \vdots\\
	                                 0 & 1 & 1
	 \end{pmatrix}\in\R^{\nni\times 3}, \quad
	 \bbi = \begin{pmatrix}b^1_1\\ b^1_2\\ b^2\end{pmatrix} \sim\mathcal{N}_3(\bo, \textrm{diag}(1.5^2, 1.8^2, 1.8^2)), 
\end{displaymath}
$\epsi\sim\mathcal{N}_{\nni}(\bo, \sigmazero^2\one)$
for $i\in\indset{\NN}$, and
$\bbi\independent\bj$, $\bbi\independent(\epsi, \epsj)$, and  $\epsi\independent\epsj$ for $i\neq j$, $i, j\in\indset{\NN}$, where the first column of $\Zi$ consists of $\floor*{0.5\nni}$ entries of $1$'s and $\ceil*{0.5\nni}$ entries of $0$'s and correspondingly for the second column of $\Zi$.

\section{Assumptions and Additional Definitions}\label{sect:AssumptionsDefinitions}

Recall the partially linear mixed-effects model 
\begin{displaymath}
	\Yi = \Xi\betazero + \g(\Wi) + \Zi\bbi+ \epsi, \quad i\in\indset{\NN}
\end{displaymath}
for groups $i\in\indset{\NN}$ as in~\eqref{eq:PLMM}. 
We consider $\NN$ grouped observations $\{\Si=(\Wi,\Xi, \Zi, \Yi)\}_{i\in\indset{\NN}}$ from this model that satisfy Assumption~\ref{assumpt:Distribution}. 
In each group $i\in\indset{\NN}$, we observe $\nni$ observations. We assume that these numbers are uniformly bounded by $\nmax<\infty$, that is, $\nni\le\nmax$ for all $i\in\indset{\NN}$. We denote the total number of observations of all groups by $\NNtot := \sum_{i=1}^{\NN}\nni$.\\

Let the number of sample splits $\KK\ge 2$ be a fixed integer independent of $\NN$. We assume that $\NN\ge\KK$ holds. Consider a partition $I_1,\ldots,I_{\KK}$ of $\indset{\NN}$. 
For $\kk\in\indset{\KK}$, we denote by $\nnktot:=\sum_{i\in\Ik}\nni$ the total number of observations of all groups $i$ belonging to $\Ik\subset\indset{\NN}$.
The sets $\I_1,\ldots,I_{\KK}$ are assumed to be of approximately equal size in the sense that $\KK\nnktot=\NNtot +o(1)$ holds for all $\kk\in\indset{\KK}$ as $\NN\rightarrow\infty$, 
which implies $\frac{\NNtot}{\nnktot} = O(1)$. 
Moreover, we assume that  $\frac{\normone{\Ik}}{\nnktot} = O(1)$ 
holds for all $\kk\in\indset{\KK}$.

For $\kk\in\indset{\KK}$, denote by $\SIkc := \{\Si\}_{i\in\Ikc}$ the grouped observations from $\Ikc$. We denote the nuisance parameter estimator that is estimated with data from $\Ikc$ by $\hetaIkc=\hetaIkc(\SIkc)$. 

\begin{definition}
For $\kk\in\indset{\KK}$, $\theta\in\Theta$, and $\eta\in\TauN$, where $\Theta$ and $\TauN$ are defined in Assumptions~\ref{assumpt:Theta} and~\ref{assumpt:DML}, respectively,  we introduce the notation
\begin{displaymath}
	\Enk{\psi (\S; \theta,\eta)} := \frac{1}{\nnktot}\sum_{i\in\Ik}\psi (\Si; \theta,\eta).
\end{displaymath}
\end{definition}

Let $\{\deltaN\}_{\NN\ge \KK}$ and $\{\DeltaN\}_{\NN\ge \KK}$ be two sequences of non-negative numbers that converge to $0$ as $\NN\rightarrow\infty$, where $\deltaN^2\ge \NN^{-\frac{1}{2}}$ holds. 
We assume that $\normone{\Ik}^{-\frac{1}{2}+\frac{1}{p}}\log(\normone{\Ik})\lesssim\deltaN$
holds for all $\kk\in\indset{\KK}$, where $p$ is specified in Assumption~\ref{assumpt:regularity}. 
Let $\{\PcalN\}_{\NN\ge 1}$ be a sequence of sets of probability distributions $\PP$ of the $\NN$ grouped observations $\{\Si=(\Wi,\Xi,\Yi)\}_{i\in\indset{\NN}}$. 
We make the following additional assumptions. 

\begin{assumptions}\label{assumpt:regularity}
Let $p\ge 8$. 
	For all $\NN$, all $i\in\indset{\NN}$,  all $\PP\in\PcalN$, and all $\kk\in\indset{\KK}$, we have the following. 
	\begin{enumerate}[label={\theassumptions.\arabic*}]
		\item\label{assumpt:regularity1}
			At the true  $\thetazero$ and the true  $\etazero$, 
			the data $\{\Si=(\Wi,\Xi,\Zi,\Yi)\}_{i\in\indset{\NN}}$ satisfies the identifiability condition  $\EPbig{\Enk{\psi(\S;\thetazero,\etazero)}} = \bo$.
	
		\item\label{assumpt:regularity2}
		There exists a finite real constant $\CpnormRV$ satisfying $\normP{\Xi}{p}+\normP{\Yi}{p}\le \CpnormRV$ for all $i\in\indset{\NN}$.
		
		\item\label{assumpt:regularity6}
		The matrices $\Zi$ assigning the random effects inside a group are fixed and  bounded. In particular, there exists a finite real constant $\CboundZi$ satisfying 
		$\norm{\Zi}\le\CboundZi$ for all $i\in\indset{\NN}$.
		
		\item\label{assumpt:regularity4}
		        In absolute value, the smallest and largest singular values of the Jacobian matrix 
		        \begin{displaymath}
		        		J_0 := \partial_{\theta}\EPBig{\Enkbig{\psi(\S;\theta,\etazero)}}\Big|_{\theta=\theta_0}
		        \end{displaymath}
	are bounded away from $0$ by $\cone>0$ and are bounded away from $+\infty$ by $\ctwo<\infty$. 	
	
		\item\label{assumpt:regularity5}
			For all $\theta\in\Theta$, we have the identification condition
			\begin{displaymath}
				\min\{\norm{J_0(\theta-\thetazero)}, \cone\} \le 2\normBig{\EPBig{\Enkbig{\psi(\S;\theta,\etazero)}}}.
			\end{displaymath}
			
		\item\label{assumpt:regularity3}
		The matrix $\EP[\RXi^T(\Zi\Sigmazero\Zi+\sigmazero^2\oneni)^{-1}\RXi]\in\R^{d\times d}$ exists and is invertible for all $i\in\indset{\NN}$. We assume that the same holds if $\thetazero$ and $\etazero$ are replaced by $\theta\in\Theta$ and $\eta\in\TauN$, respectively, with $\Theta$ as in Assumption~\ref{assumpt:Theta} and $\TauN$ as in Assumption~\ref{assumpt:DML}. 
		
		\item\label{assumpt:regularity3-2}
			The singular values of the symmetric matrix $\EP[\RXi^T(\Zi\Sigmazero\Zi+\sigmazero^2\oneni)^{-1}\RXi]\in\R^{d\times d}$ are uniformly bounded away from $0$ by $\clambdamin>0$ for all $i\in\indset{\NN}$.
			
		\item\label{assumpt:regularity7}
			There exists a symmetric positive-definite matrix $\Tzero\in\R^{d\times d}$ satisfying
			\begin{displaymath}
				\Tbar := \frac{1}{\NNtot}\sum_{i=1}^{\NN}\EPbig{\RXi^T\Vzeroi^{-1}\RXi}
	= \Tzero + o(1).
			\end{displaymath}
	\end{enumerate}
\end{assumptions}

Assumption~\ref{assumpt:regularity1} ensures that $\betazero$ is identifiable by our estimation method. 
Assumption~\ref{assumpt:regularity2} ensures that enough moments of $\Xi$ and $\Yi$ exist. 
Assumption~\ref{assumpt:regularity4} and~\ref{assumpt:regularity5} are required to prove that $\thetazero$ is consistently estimated in Lemma~\ref{lem:consistency}. The proof of this lemma  uses a Taylor expansion. 
Assumption~\ref{assumpt:regularity3},~\ref{assumpt:regularity3-2}, and~\ref{assumpt:regularity7} are required to make statements about the asymptotic variance-covariance matrix in the proof of Theorem~\ref{thm:asymptoticGaussian}.

The following Assumption~\ref{assumpt:Theta} characterizes the set $\Theta$ to which $\thetazero$ belongs and from which estimators of $\thetazero$ are not too far away in the sense of Assumption~\ref{assumpt:Theta2}.

\begin{assumptions}\label{assumpt:Theta}
Consider the set  
\begin{displaymath}
	\Theta := \big\{\theta =(\beta, \Sigma, \sigma^2)\in\R^d \times \R^{q\times q}\times \R\colon \Sigma\in\R^{q\times q} \textrm{ symmetric positive definite, } \sigma>0\big\}
\end{displaymath}
of parameters. We make the following assumptions on $\Theta$ and $\hthetak$ for $\kk\in\indset{\KK}$.
\begin{enumerate}[label={\theassumptions.\arabic*}]
	\item\label{assumpt:Theta1}
		The set $\Theta$ is bounded and contains $\thetazero$ and a ball of radius $\max_{\NN\ge 1}\deltaN$ 
		around $\thetazero$. 
	
	\item\label{assumpt:Theta3}
	There exists a finite real constant $\CboundV$ such that we have 
	$\norm{(\Zi\Sigma\Zi^T + \oneni )^{-1}}\le \CboundV$
	for all $i\in\indset{\NN}$ and all $\Sigma$ belonging to $\Theta$.
		
	\item\label{assumpt:Theta2}
		For all $\kk\in\indset{\KK}$,
		the estimator $\hthetak$ belongs to $\Theta$ and 
		satisfies the approximate solution property
\begin{displaymath}
	\normbig{\Enkbig{\psibig{\S; \hthetak,\hetaIkc}}} 
	\le \inf_{\theta\in\Theta} \normbig{\Enkbig{\psibig{\S; \theta,\hetaIkc}}} + \errorN
\end{displaymath}
	with the nuisance parameter estimator $\hetaIkc=\hetaIkc(\SIkc)$,
where $\{\errorN\}_{\NN\ge \KK}$ is a sequence of non-negative numbers satisfying $\errorN\lesssim\deltaN^2$.
\end{enumerate}
\end{assumptions}

The following Assumption~\ref{assumpt:DML} mainly characterizes the $\NN^{-1/2}$ product convergence rate of the machine learners that estimate the conditional expectations, which are nuisance functions.

\begin{assumptions}\label{assumpt:DML}
	Consider the $p\ge 8$ from Assumption~\ref{assumpt:regularity}.  
	For all $\NN\ge \KK$ and all $\PP\in\PcalN$, 
	consider a nuisance function realization set $\TauN$ such that the following conditions hold.
	\begin{enumerate}[label={\theassumptions.\arabic*}]
		\item\label{assumpt:DML1}
		The set $\TauN$ consists of $\PP$-integrable functions $\eta=(\mX,\mY)$ whose $p$th moment exists, and it contains $\etazero$.  Furthermore, there exists a finite real constant $\CpnormEta$ such that
		\begin{displaymath}
			\begin{array}{l}
			\normP{\etazero-\eta}{p}\le \CpnormEta, 
			\quad
			\normP{\etazero-\eta}{2}\le \deltaN^8,\\
			\normP{\mX^0(W)-\mX(W)}{2}\big(\normP{\mY^0(W)-\mY(W)}{2}+\normP{\mX^0(W)-\mX(W)}{2}\big)\le\deltaN\NN^{-\frac{1}{2}}
			\end{array}
		\end{displaymath}
		hold for all elements $\eta$ of $\TauN$. 
		
		\item\label{assumpt:DML2}
		For all $\kk\in\indset{\KK}$,  the nuisance parameter estimate $\hetaIkc=\hetaIkc(\SIkc)$ satisfies 
		\begin{displaymath}
			\begin{array}{l}
			\normP{\etazero-\hetaIkc}{p}\le \CpnormEta,
			\quad
			\normP{\etazero-\hetaIkc}{2}\le \deltaN^8, \\
			\normP{\mX^0(W)-\hmX^{\Ikc}(W)}{2}\big(\normP{\mY^0(W)-\hmY^{\Ikc}(W)}{2}+ \normP{\mX^0(W)-\hmX^{\Ikc}(W)}{2} \big)\le\deltaN\NN^{-\frac{1}{2}}
			\end{array}
		\end{displaymath}
		with $\PP$-probability no less than $1-\DeltaN$.
		Denote by $\EpsN$ the event that $\hetaIkc=\hetaIkc(\SIkc)$, $\kk\in\indset{\KK}$ belong to $\TauN$, and assume this event  holds with $\PP$-probability at least 
		 $1-\DeltaN$. 
		
		\item\label{assumpt:DML3}
		For all $\kk\in\indset{\KK}$, the parameter estimator $\hthetak$ is $\PP$-integrable and its $p$th moment exists. 
	\end{enumerate}
\end{assumptions}

We suppose all assumptions presented in Section~\ref{sect:AssumptionsDefinitions} of the appendix hold throughout the remainder of the appendix.

\section{Proof of Theorem~\ref{thm:asymptoticGaussian}}\label{sect:proofAsymptoticGaussian}

\subsection{Supplementary Lemmata}

\begin{lemma}\citep[Lemma G.7]{Emmenegger2021}\label{lem:Emmenegger2021-Lemma-G-7}
	Let $\ttt\ge 1$. 
	Consider a $t$-dimensional random variable $A$ and an $s$-dimensional random variable $B$. Denote the joint law of $A$ and $B$ by $\PP$. Then, we have 
	\begin{displaymath}
		\normP{A-\EP[A|B]}{{\ttt}}\le 2\normP{A}{\ttt}. 
	\end{displaymath}
\end{lemma}

\begin{lemma}\label{lem:EmmeneggerG10} \citep[Lemma G.10]{Emmenegger2021}\label{lem:Emmenegger2021-Lemma-G-10}
	Consider a $t_1$-dimensional random variable $A_1$, a $t_2$-dimensional  random variable $A_2$, and an $s$-dimensional random variable $B$. 
	Denote the joint law of $A_1$, $A_2$, and $B$ by $\PP$. Then, we have 
	\begin{displaymath}
		\normbig{\EP\big[(A_1-\EP[A_1|B])A_2^T\big]}^2\le \normP{A_1}{2}^2\normP{A_2}{2}^2.
	\end{displaymath}
\end{lemma}

The following lemma, proved in~\citet[Lemma 6.1]{Chernozhukov2018} and~\citet[Lemma G.12]{Emmenegger2021}, states that conditional convergence in probability implies unconditional convergence in probability. 

\begin{lemma}\label{lem:ChernozhukovLemma}(\citet[Lemma 6.1]{Chernozhukov2018}; \citet[Lemma G.12]{Emmenegger2021})
Let $\{A_n\}_{n\ge 1}$ and $\{B_n\}_{n\ge 1}$ be sequences of random vectors, and let ${\ttt}\ge 1$.
Consider a deterministic sequence $\{\eps_n\}_{n\ge 1}$ with $\eps_n\rightarrow 0$ as $n\rightarrow\infty$ such that $\E[\norm{A_n}^{\ttt} |B_n]\le\eps_n^{\ttt}$ holds. 
Then, we have $\norm{A_n}=O_{\PP}(\eps_n)$ unconditionally, meaning that that for any sequence $\{\ell_n\}_{n\ge 1}$ with $\ell_n\rightarrow\infty$ as $n\rightarrow\infty$, we have $\PP(\norm{A_n}>\ell_n\eps_n)\rightarrow 0$.
\end{lemma}

\subsection{Representation of the Score Function $\psi$}

\begin{lemma}\label{lem:psiDerivativeBeta}
	Let $i\in\indset{\NN}$, $\theta\in\Theta$, and $\eta\in\TauN$. 
	Denote by $\Vi := \Zi\Sigma\Zi^T+\oneni$. Furthermore,  
	denote by $\psibeta$ the coordinates of $\psi$ that correspond to $\beta$, that is, 
	$\psibeta(\Si;\theta,\eta) = \nabla_{\beta}\li(\theta,\eta)$. 
	We have
	\begin{displaymath}
		\psibeta(\Si; \theta,\eta) = \frac{1}{\sigma^2} \big(\Xi - \mX(\Wi)\big)^T\Vi^{-1} 
			\Big( \Yi - \mY(\Wi)- \big(\Xi - \mX(\Wi)\big)\beta\Big).
	\end{displaymath}
\end{lemma}
\begin{proof}
The statement follows from the definition of $\psi$. 
\end{proof}

\begin{lemma}\label{lem:psiDerivativeSigma2}
	Let $i\in\indset{\NN}$, $\theta\in\Theta$, and $\eta\in\TauN$. 
	Denote by $\Vi := \Zi\Sigma\Zi^T+\oneni$. Furthermore,  
	denote by $\psisigma$ the coordinates of $\psi$ that correspond to $\sigma^2$, that is, 
	$\psisigma(\Si;\theta,\eta) = \nabla_{\sigma^2}\li(\theta,\eta)$. 
	We have
	\begin{displaymath}
	\begin{array}{rl}
		&\psisigma(\Si; \theta,\eta) \\
		=& -\frac{\nni}{2\sigma^2} + 
		\frac{1}{2(\sigma^2)^2} \Big( \Yi - \mY(\Wi)- \big(\Xi - \mX(\Wi)\big)\beta \Big)^T\Vi^{-1} 
			\Big( \Yi - \mY(\Wi)- \big(\Xi - \mX(\Wi)\big)\beta\Big).
	\end{array}
	\end{displaymath}
\end{lemma}
\begin{proof}
The statement follows from the definition of $\psi$. 
\end{proof}

\begin{lemma}\label{lem:psiDerivativeSigmatilkl}
	Let $i\in\indset{\NN}$, $\theta\in\Theta$,  $\eta\in\TauN$. 
	Denote by $\Vi := \Zi\Sigma\Zi^T+\oneni$. 
	Furthermore, let indices $\kappa, \iota\in\indset{q}$, and 
	denote by $\psiSigmakj$ the coordinates of $\psi$ that correspond to $\Sigmakj$, that is, 
	$\psiSigmakj(\Si;\theta,\eta) = \nabla_{\Sigmakj}\li(\theta,\eta)$. 
	We have
	\begin{displaymath}
	\begin{array}{rl}
		&\psiSigmakj(\Si; \theta,\eta) \\
		=& -\frac{1}{2}  \sum_{t,u=1}^{\nni}(\Vi^{-1})_{t,u}(\Zi)_{t,\kappa}(\Zi^T)_{\iota,u}\\
		&\quad
		+\frac{1}{2\sigma^2} \sum_{t,u=1}^{\nni}  \Big(\Yi - \mY(\Wi)- \big(\Xi - \mX(\Wi)\big)\beta\Big)_t \Big(\Yi - \mY(\Wi)- \big(\Xi - \mX(\Wi)\big)\beta\Big)_u \\
		&\quad\quad\quad\quad\quad\quad\quad\cdot \big(\Vi^{-1}(\Zi)_{\cdot,\kappa}(\Zi^T)_{\iota,\cdot}\Vi^{-1}\big)_{t,u}.
	\end{array}
	\end{displaymath}
\end{lemma}
\begin{proof}
Let a vector $x\in\R^{\nni}$. We have
\begin{displaymath}
	\begin{array}{rl}
	 & \frac{\partial}{\partial\Sigmakj} \Big(x^T\big(\Zi\Sigma\Zi^T+\oneni\big)^{-1}x\Big)\\
	 =& \sum_{t,u=1}^{\nni} \frac{\partial}{\partial(\Zi\Sigma\Zi^T+\oneni)^{-1}_{t,u}} \Big(x^T\big(\Zi\Sigma\Zi^T+\oneni\big)^{-1}x\Big) 
	 \cdot \frac{\partial(\Zi\Sigma\Zi^T+\oneni)^{-1}_{t,u}}{\partial\Sigmakj}.
	 \end{array}
\end{displaymath}
For some nonrandom matrix $D\in\R^{\nni\times\nni}$, we have
\begin{displaymath}
	\frac{\partial}{\partial D_{t,u}} x^TDx = \frac{\partial}{\partial D_{t,u}} \sum_{r,s=1}^{\nni} x_rD_{r,s}x_s = x_tx_u. 
\end{displaymath}
Furthermore, we have
\begin{displaymath}
	\frac{\partial}{\partial\Sigmakj}\big(\Zi\Sigma\Zi^T+\oneni\big)^{-1} 
	=-\big(\Zi\Sigma\Zi^T+\oneni\big)^{-1} \bigg(\frac{\partial}{\partial\Sigmakj}\big(\Zi\Sigma\Zi^T+\oneni\big)\bigg)  \big(\Zi\Sigma\Zi^T+\oneni\big)^{-1}
\end{displaymath}
by~\citet[Equation (59)]{Petersen-Pedersen2012}, and we have
\begin{displaymath}
	\bigg(\frac{\partial}{\partial\Sigmakj}\big(\Zi\Sigma\Zi^T+\oneni\big)\bigg)_{t,u}
	= \frac{\partial}{\partial\Sigmakj} \sum_{r,s=1}^{\nni} (\Zi)_{t,r}\Sigma_{r,s}(\Zi^T)_{s,u}=(\Zi)_{t,\kappa}(\Zi^T)_{\iota,u},
\end{displaymath}
and consequently
\begin{displaymath}
	\frac{\partial}{\partial\Sigmakj}\big(\Zi\Sigma\Zi^T+\oneni\big)
	= (\Zi)_{\cdot,\kappa}(\Zi^T)_{\iota,\cdot}, 
\end{displaymath}
which leads to
\begin{displaymath}
	\frac{\partial}{\partial\Sigmakj}\big(\Zi\Sigma\Zi^T+\oneni\big)^{-1} 
	=-\big(\Zi\Sigma\Zi^T+\oneni\big)^{-1} (\Zi)_{\cdot,\kappa}(\Zi^T)_{\iota,\cdot} \big(\Zi\Sigma\Zi^T+\oneni\big)^{-1}.
\end{displaymath}
Therefore, we have
\begin{equation}~\label{eq:diffSigmakl1}
	\begin{array}{rl}
	 & \frac{\partial}{\partial\Sigmakj} \Big(x^T\big(\Zi\Sigma\Zi^T+\oneni\big)^{-1}x\Big)\\
	 =&- \sum_{t,u=1}^{\nni} x_tx_u
	 \cdot \Big(\big(\Zi\Sigma\Zi^T+\oneni\big)^{-1} (\Zi)_{\cdot,\kappa}(\Zi^T)_{\iota,\cdot} \big(\Zi\Sigma\Zi^T+\oneni\big)^{-1}\Big)_{t,u}. 
	 \end{array}
\end{equation}

Moreover, we have
\begin{equation}~\label{eq:diffSigmakl2}
\begin{array}{rl}
	&\frac{\partial}{\partial\Sigmakj} \log\Big(\det \big(\Zi\Sigma\Zi^T+\oneni\big)\Big)\\
	=& \sum_{t,u=1}^{\nni} \frac{\partial}{\partial(\Zi\Sigma\Zi^T+\oneni)_{t,u}} \log\Big(\det \big(\Zi\Sigma\Zi^T+\oneni\big)\Big) \cdot \frac{\partial(\Zi\Sigma\Zi^T+\oneni)_{t,u}}{\partial\Sigmakj}\\
	=& \sum_{t,u=1}^{\nni} \Big(\big(\Zi\Sigma\Zi^T+\oneni\big)^{-1}\Big)_{t,u}(\Zi)_{t,\kappa}(\Zi^T)_{\iota,u}
\end{array}
\end{equation}
by~\citet[Equation (57)]{Petersen-Pedersen2012}. 
We replace $x$ in~\eqref{eq:diffSigmakl1} by $\Yi - \mY(\Wi)- \big(\Xi - \mX(\Wi)\big)\beta $ and combine~\eqref{eq:diffSigmakl1} and~\eqref{eq:diffSigmakl2} to conclude the proof. 
\end{proof}

\subsection{Consistency}

This section establishes that all $\hthetak$, $\kk\in\indset{\KK}$ are consistent. In particular, this implies that $\htheta$ is consistent. 
\\

Let $\PP\in\PcalN$.

\begin{lemma}\label{lem:consistency}
	Let $\kk\in\indset{\KK}$. 
	We have $\norm{\hthetak-\thetazero} \le \deltaN^2$ with $\PP$-probability $1-o(1)$. 
\end{lemma}
\begin{proof}[Proof of Lemma~\ref{lem:consistency}]
We have
\begin{equation}\label{eq:lemConsistentEq1}
	\begin{array}{rl}
		&\EPBig{\Enkbig{\psibig{\S; \hthetak,\etazero}}} \\
		=& \EPBig{\Enkbig{\psibig{\S; \hthetak,\etazero}} - \Enkbig{\psibig{\S; \hthetak,\hetaIkc}}} \\
		&\quad
		+ \EPBig{\Enkbig{\psibig{\S; \hthetak,\hetaIkc}}} - \Enkbig{\psibig{\S; \hthetak,\hetaIkc}} + \Enkbig{\psibig{\S; \hthetak,\hetaIkc}}. 
	\end{array}
\end{equation}
Due to the approximate solution property in Assumption~\ref{assumpt:Theta2}, the identifiability condition in Assumption~\ref{assumpt:regularity1}, and the triangle inequality, we have
\begin{equation}\label{eq:lemConsistentEq2}
	\begin{array}{rl}
		&\normbig{\Enkbig{\psibig{\S; \hthetak,\hetaIkc}}} \\
		\le & \normbig{\Enkbig{\psibig{\S; \thetazero,\hetaIkc}}} + \errorN\\
		\le &  \normBig{\Enkbig{\psibig{\S;\thetazero,\hetaIkc}} - \EPBig{\Enkbig{\psibig{\S;\thetazero,\hetaIkc}}}}\\
		&\quad + \normBig{\EPBig{\Enkbig{\psibig{\S;\thetazero,\hetaIkc}}} - \EPBig{\Enkbig{\psibig{\S;\thetazero,\etazero}}}} + \errorN.
	\end{array}
\end{equation}
Let us introduce
\begin{equation}\label{eq:Icalone}
	\Icalone := \sup_{\substack{\theta\in\Theta, \\ \eta\in\TauN}} 
	\normBig{\EPBig{\Enkbig{\psi(\S;\theta,\eta)}} - \EPBig{\Enkbig{\psibig{\S;\theta,\etazero}}}}
\end{equation}
and 
\begin{equation}\label{eq:Icaltwo}
	\Icaltwo := \sup_{\theta\in\Theta} \normBig{\Enkbig{\psibig{\S;\theta,\hetaIkc}} - \EPBig{\Enkbig{\psibig{\S;\theta,\hetaIkc}}}}.
\end{equation}
Due to~\eqref{eq:lemConsistentEq1} and~\eqref{eq:lemConsistentEq2}, we infer, with $\PP$-probability $1-o(1)$, 
\begin{displaymath}
	\normBig{\EPBig{\Enkbig{\psibig{\S; \hthetak,\etazero}}} }
	\le \errorN + 2\Icalone + 2\Icaltwo
\end{displaymath}
because the event $\EpsN$ that $\hetaIkc$ belongs to $\TauN$ holds with $\PP$-probability $1-o(1)$ by Assumption~\ref{assumpt:DML2}. 
By Lemma~\ref{lem:boundIcalone}, we have $\Icalone\lesssim\deltaN^2$. 
By Lemma~\ref{lem:boundIcaltwo}, we have $\Icaltwo\lesssim\NN^{-\frac{1}{2}}$ with $\PP$-probability $1-o(1)$.
Recall that we have $\deltaN^2\ge\NN^{-\frac{1}{2}}$ and $\errorN\lesssim\deltaN^2$. 
With $\PP$-probability $1-o(1)$, we therefore have
\begin{displaymath}
	\min\{\norm{J_0(\hthetak-\thetazero)}, \cone\} \le 2\normBig{\EPBig{\Enkbig{\psi(\S;\hthetak,\etazero)}}} \lesssim\deltaN^2
\end{displaymath}
due to Assumption~\ref{assumpt:regularity5}. We infer our claim because the singular values of $\Jzero$ are bounded away from $0$ by Assumption~\ref{assumpt:regularity4}. 
\end{proof} 

\begin{lemma}\label{lem:boundIcalone}
	Consider
	\begin{displaymath}
	\Icalone = \sup_{\substack{\theta\in\Theta, \\ \eta\in\TauN}} 
	\normBig{\EPBig{\Enkbig{\psi(\S;\theta,\eta)}} - \EPBig{\Enkbig{\psibig{\S;\theta,\etazero}}}}
\end{displaymath}
	as in~\eqref{eq:Icalone}. 
	We have $\Icalone\lesssim\deltaN^2$. 
\end{lemma}
\begin{proof}
	Let indices $i\in\indset{\NN}$ and $\kappa,\iota\in\indset{q}$, let $\theta\in\Theta$, and let $\eta\in\TauN$. 
	 Furthermore,
	let $\psibeta(\Si;\theta,\eta) := \nabla_{\beta}\li(\theta,\eta)$, 
	let $\psisigma(\Si;\theta,\eta) := \nabla_{\sigma^2}\li(\theta,\eta)$, and  
	let $\psiSigmakj(\Si;\theta,\eta) := \nabla_{\Sigmakj}\li(\theta,\eta)$. 
	Denote by $\Vi := \Zi\Sigma\Zi^T+\one_{\nni}$.
	We have
	\begin{equation}\label{eq:boundIcalone:eq1}
		\begin{array}{rl}
			& \psibeta(\Si; \theta,\eta) - \psibetabig{\Si;\theta,\etazero}\\
			=& \frac{1}{\sigma^2}\big(\Xi - \mX^0(\Wi)\big)^T\Vi^{-1} \Big( \mY^0(\Wi) - \mY(\Wi) - \big(\mX^0(\Wi)-\mX(\Wi)\big)\beta \Big)\\
			&\quad + \frac{1}{\sigma^2} \big(\mX^0(\Wi)-\mX(\Wi)\big)^T\Vi^{-1} \Big( \Yi - \mY^0(\Wi)  - \big(\Xi - \mX^0(\Wi)\big)\beta \Big)\\
			&\quad + \frac{1}{\sigma^2} \big(\mX^0(\Wi)-\mX(\Wi)\big)^T\Vi^{-1}\Big( \mY^0(\Wi) - \mY(\Wi) - \big(\mX^0(\Wi)-\mX(\Wi)\big)\beta \Big),
		\end{array}
	\end{equation}
we have
\begin{equation}\label{eq:boundIcalone:eq2}
	\begin{array}{rl}
		& \psisigma(\Si; \theta,\eta) - \psisigmabig{\Si;\theta,\etazero}\\
		=& 2 \cdot \frac{1}{2(\sigma^2)^2} \Big( \Yi - \mY^0(\Wi)  - \big(\Xi - \mX^0(\Wi)\big)\beta \Big)^T\Vi^{-1}\\
		&\quad\quad\cdot\Big( \mY^0(\Wi) - \mY(\Wi) - \big(\mX^0(\Wi)-\mX(\Wi)\big)\beta \Big)\\
		&\quad + \frac{1}{2(\sigma^2)^2} \Big( \mY^0(\Wi) - \mY(\Wi) - \big(\mX^0(\Wi)-\mX(\Wi)\big)\beta \Big)^T\Vi^{-1}\\
		&\quad\quad\cdot\Big( \mY^0(\Wi) - \mY(\Wi) - \big(\mX^0(\Wi)-\mX(\Wi)\big)\beta \Big), 
	\end{array}
\end{equation}
and we have
\begin{equation}\label{eq:boundIcalone:eq3}
	\begin{array}{rl}
		& \psiSigmakj(\Si; \theta,\eta) - \psiSigmakjbig{\Si;\theta,\etazero}\\
		=& \frac{1}{2\sigma^2} \sum_{t,u=1}^{\nni} 
		\big(\Vi^{-1}(\Zi)_{\cdot,\kappa}(\Zi^T)_{\iota,\cdot}\Vi^{-1} \big)_{t,u}\\
		&\quad\cdot\bigg( \Big(\mY^0(\Wi) - \mY(\Wi) - \big(\mX^0(\Wi)-\mX(\Wi)\big)\beta \Big)_t \Big( \Yi - \mY^0(\Wi)  - \big(\Xi - \mX^0(\Wi)\big)\beta \Big)_u\\
		&\quad \quad+ \Big( \Yi - \mY^0(\Wi)  - \big(\Xi - \mX^0(\Wi)\big)\beta \Big)_t\Big(\mY^0(\Wi) - \mY(\Wi) - \big(\mX^0(\Wi)-\mX(\Wi)\big)\beta \Big)_u\\
		&\quad\quad +  \Big(\mY^0(\Wi) - \mY(\Wi) - \big(\mX^0(\Wi)-\mX(\Wi)\big)\beta \Big)_t\\
		&\quad\quad\quad\quad\cdot\Big(\mY^0(\Wi) - \mY(\Wi) - \big(\mX^0(\Wi)-\mX(\Wi)\big)\beta \Big)_u
		\bigg).
	\end{array}
\end{equation}
Up to constants depending on the diameter of $\Theta$, the $L^1$-norms of all terms~\eqref{eq:boundIcalone:eq1}--\eqref{eq:boundIcalone:eq3} are bounded by $\deltaN$ due to H\"{o}lder's inequality
because we have $\nni\le\nmax$, $\normP{\Xi - \mX^0(\Wi)}{2}\le\normP{\Xi}{2}$ by Lemma~\ref{lem:EmmeneggerG10} and similarly for $\Yi$, 
$\normP{\Xi}{2}$ and $\normP{\Yi}{2}$ are bounded by Assumption~\ref{assumpt:regularity2} and H\"{o}lder's inequality, 
$\Zi$ is bounded by Assumption~\ref{assumpt:regularity6}, 
$\Vi^{-1}=(\Zi\Sigma\Zi^T + \oneni)^{-1}$ is bounded by Assumption~\ref{assumpt:Theta3},
$\normP{\etazero-\eta}{2}\le\deltaN^8\le\deltaN^2$ holds by Assumption~\ref{assumpt:DML1} for $\NN$ large enough, 
and $\Theta$ is bounded by Assumption~\ref{assumpt:Theta1}. 
Therefore, we infer the claim. 
\end{proof}

\begin{lemma}\label{lem:exampleVaart}
	Let $\eta\in\TauN$, and consider the function class $\Fcaleta:= \{\psi_j(\cdot; \theta,\eta)\colon j\in\indset{d+1+q^2}, \theta\in\Theta\}$.
	Let $i\in\indset{\NN}$ and $\theta_1,\theta_2\in\Theta$. Then, there exists a function $h\in L^2$ such that for all $f_{\theta_1}, f_{\theta_2}\in\Fcaleta$, we have
	\begin{displaymath}
		\normone{f_{\theta_1}(\cdot) - f_{\theta_2}(\cdot)} \le h(\cdot) \norm{\theta_1-\theta_2}.
	\end{displaymath}
\end{lemma}
\begin{proof}
Let $i\in\indset{\NN}$, and consider the grouped observations $\Si$ of group $i$. Independently of $i$, the number of observations $\nni$ from this group is bounded by $\nmax<\infty$. 

Let $\eta=(\mX,\mY)\in\TauN$, and let $\theta_1,\theta_2\in\Theta$. 
Denote by $\Vione := \Zi\Sigma_1\Zi^T+\oneni$, and denote by $\Vitwo := \Zi\Sigma_2\Zi^T+\oneni$. 
Moreover, denote by $\RXieta := \Xi-\mX(\Wi)$ and by $\RYieta := \Yi-\mY(\Wi)$.
Furthermore, consider indices $\kappa,\iota\in\indset{q}$, and 
	let $\psibeta(\Si;\theta,\eta) := \nabla_{\beta}\li(\theta,\eta)$, 
	let $\psisigma(\Si;\theta,\eta) := \nabla_{\sigma^2}\li(\theta,\eta)$, and  
	let $\psiSigmakj(\Si;\theta,\eta) := \nabla_{\Sigmakj}\li(\theta,\eta)$. 
Observe that
\begin{equation}\label{eq:rewriteDiffMatInv}
	\Vione^{-1} - \Vitwo^{-1} = \Vione^{-1}\big(\Vitwo-\Vione\big)\Vitwo^{-1}
\end{equation}
and
\begin{equation}\label{eq:rewriteDiffMatInv-2}
	\frac{1}{\sigma_1^2} - \frac{1}{\sigma_2^2} = \big( \sigma_1^2 \big)^{-1} \big(\sigma_2^2-\sigma_1^2\big) \big( \sigma_2^2 \big)^{-1}
\end{equation}
hold. Thus, we have 
\begin{displaymath}
		\begin{array}{rl}
			& \psibeta(\Si; \theta_1,\eta) - \psibetabig{\Si;\theta_1,\eta}\\
			=&\Big(\frac{1}{\sigma_1^2}-\frac{1}{\sigma_2^2} \Big) \RXieta^T\Vione^{-1}\big(\RYieta-\RXieta^T\beta_1\big)
			+\frac{1}{\sigma_2^2}\RXieta^T \Vione^{-1}\big(\Vitwo-\Vione\big)\Vitwo^{-1}\RYieta\\
			&\quad -\frac{1}{\sigma_2^2} \RXieta^T\Big( \Vione^{-1}\RXieta(\beta_1-\beta_2) + \Vione^{-1} \big(\Vitwo-\Vione\big)\Vitwo^{-1}\RXieta\beta_2 \Big), 
		\end{array}
	\end{displaymath}
and 
\begin{displaymath}
	\begin{array}{rl}
		& \psisigma(\Si; \theta_1,\eta) - \psisigmabig{\Si;\theta_2, \eta}\\
		=& \frac{\nni}{2}\Big(\frac{1}{\sigma_2^2}-\frac{1}{\sigma_1^2}\Big) + \frac{1}{2(\sigma_1^2)^2}\big(\RYieta-\RXieta\beta_1\big)^T\Vione^{-1}\big(\Vitwo-\Vione\big)\Vitwo^{-1}\big(\RYieta-\RXieta\beta_1\big)\\
		&\quad + \frac{1}{2} \Big(\frac{1}{\sigma_1^2}-\frac{1}{\sigma_2^2}\Big) \big(\RYieta-\RXieta\beta_1\big)^T\Vitwo^{-1} \big(\RYieta-\RXieta\beta_1\big)\\
		&\quad + \frac{2}{2(\sigma_2^2)^2} \big(\RYieta-\RXieta\beta_2\big)\Vitwo^{-1} \RXieta(\beta_2-\beta_1)\\
		&\quad + \frac{1}{2(\sigma_2^2)^2}(\beta_2-\beta_1)^T\RXieta^T\Vitwo^{-1} \RXieta(\beta_2-\beta_1), 
	\end{array}
\end{displaymath}
and 
\begin{displaymath}
	\begin{array}{rl}
		& \psiSigmakj(\Si; \theta_1,\eta) - \psiSigmakjbig{\Si; \theta_2,\eta}\\
		=& -\frac{1}{2} \sum_{t,u=1}^{\nni} \Big(\Vione^{-1}\big(\Vitwo-\Vione\big)\Vitwo^{-1} \Big)_{t,u} (\Zi)_{t,\kappa}(\Zi^T)_{\iota,u}\\
		&\quad + \frac{1}{2}\sum_{t,u=1}^{\nni} \Big(\frac{1}{\sigma_1^2}  \big(\RYieta-\RXieta\beta_1\big)_t \big(\RYieta-\RXieta\beta_1\big)_u \big(\Vione^{-1}(\Zi)_{\cdot,\kappa}(\Zi^T)_{\iota,\cdot} \Vione^{-1}\big)_{t,u}\\
		&\quad\quad\quad\quad\quad\quad\quad - \frac{1}{\sigma_2^2} \big(\RYieta-\RXieta\beta_2\big)_t \big(\RYieta-\RXieta\beta_2\big)_u \big(\Vitwo^{-1}(\Zi)_{\cdot,\kappa}(\Zi^T)_{\iota,\cdot} \Vitwo^{-1}\big)_{t,u} \Big),
	\end{array}
\end{displaymath}
where for $t,u\in\indset{\nni}$, we have
\begin{displaymath}
	\begin{array}{cl}
	&\frac{1}{\sigma_1^2}  \big(\RYieta-\RXieta\beta_1\big)_t \big(\RYieta-\RXieta\beta_1\big)_u \big(\Vione^{-1}(\Zi)_{\cdot,\kappa}(\Zi^T)_{\iota,\cdot} \Vione^{-1}\big)_{t,u}\\
	=& \Big( \frac{1}{\sigma_1^2} -\frac{1}{\sigma_2^2}\Big)\big(\RYieta-\RXieta\beta_1\big)_t \big(\RYieta-\RXieta\beta_1\big)_u \big(\Vione^{-1}(\Zi)_{\cdot,\kappa}(\Zi^T)_{\iota,\cdot} \Vione^{-1}\big)_{t,u}\\
	&\quad + \frac{1}{\sigma_2^2}\big(\RYieta-\RXieta\beta_1\big)_t \big(\RYieta-\RXieta\beta_1\big)_u \big(\Vione^{-1}(\Zi)_{\cdot,\kappa}(\Zi^T)_{\iota,\cdot} \Vione^{-1}\big)_{t,u}
	\end{array}
\end{displaymath}
and
\begin{displaymath}
	\begin{array}{rl}
		&\big(\RYieta-\RXieta\beta_1\big)_t \big(\RYieta-\RXieta\beta_1\big)_u \big(\Vione^{-1}(\Zi)_{\cdot,\kappa}(\Zi^T)_{\iota,\cdot} \Vione^{-1}\big)_{t,u}\\
		&\quad - \big(\RYieta-\RXieta\beta_2\big)_t \big(\RYieta-\RXieta\beta_2\big)_u \big(\Vitwo^{-1}(\Zi)_{\cdot,\kappa}(\Zi^T)_{\iota,\cdot} \Vitwo^{-1}\big)_{t,u}\\
		=& \big(\RYieta-\RXieta\beta_2\big)_t \big(\RYieta-\RXieta\beta_2\big)_u \Big( \big(\Vione^{-1}-\Vitwo^{-1}\big)_{t,\cdot}(\Zi)_{\cdot,\kappa}\big(\Vione^{-1}-\Vitwo^{-1}\big)_{u,\cdot}(\Zi)_{\cdot,\iota} \\
		&\quad\quad + \big(\Vione^{-1}-\Vitwo^{-1}\big)_{t,\cdot}(\Zi)_{\cdot,\kappa} \big( \Vitwo^{-1} \big)_{u,\cdot}(\Zi)_{\cdot,\iota} 
		+ \big( \Vitwo^{-1} \big)_{t,\cdot}(\Zi)_{\cdot,\kappa} \big(\Vione^{-1}-\Vitwo^{-1}\big)_{u,\cdot}(\Zi)_{\cdot,\iota}\Big)\\
		&\quad + \Big(  \big(\RXieta(\beta_2-\beta_1)\big)_t\big(\RYieta-\RXieta\beta_2\big)_u + \big(\RYieta-\RXieta\beta_2\big)_t\big( \RXieta(\beta_2-\beta_1) \big)_u \\
		&\quad\quad+ \big(\RXieta(\beta_2-\beta_1) \big)_t\big(\RXieta(\beta_2-\beta_1) \big)_u  \Big) 
		\big(\Vione^{-1}\big)_{t,\cdot}(\Zi)_{\cdot,\kappa} \big( \Vione^{-1} \big)_{u,\cdot}(\Zi)_{\cdot,\iota}.
	\end{array}
\end{displaymath}
Due to~\eqref{eq:rewriteDiffMatInv}, the terms $\Vione^{-1}-\Vitwo^{-1}$ can be represented in terms of $\Vitwo-\Vione$. 
Due to~\eqref{eq:rewriteDiffMatInv-2}, the terms $(\sigma_1^2)^{-1}-(\sigma_2^2)^{-1}$ can be represented in terms of $\sigma_2^2-\sigma_1^2$. 
Recall that $\nni\le\nmax$, 
$\normP{\Xi - \mX^0(\Wi)}{2}\le\normP{\Xi}{2}$ by Lemma~\ref{lem:EmmeneggerG10} and similarly for $\Yi$, 
$\normP{\Xi}{2}$ and $\normP{\Yi}{2}$ are bounded by Assumption~\ref{assumpt:regularity2} and H\"{o}lder's inequality, 
$\Zi$ is bounded by Assumption~\ref{assumpt:regularity6}, 
$\Vione^{-1}=(\Zi\Sigma_1\Zi^T + \oneni)^{-1}$ and $\Vitwo^{-1}=(\Zi\Sigma_2\Zi^T + \oneni)^{-1}$ are  bounded by Assumption~\ref{assumpt:Theta3},
$\mX$ and $\mY$ are square integrable by Assumption~\ref{assumpt:DML1}, 
and $\Theta$ is bounded by Assumption~\ref{assumpt:Theta1}.
Therefore, we infer the claim. 
\end{proof}

\begin{lemma}\label{lem:boundIcaltwo}
	Consider
	\begin{displaymath}
	\Icaltwo = \sup_{\theta\in\Theta} \normBig{\Enkbig{\psibig{\S;\theta,\hetaIkc}} - \EPBig{\Enkbig{\psibig{\S;\theta,\hetaIkc}}}}
\end{displaymath}
	as in~\eqref{eq:Icaltwo}. 
	We have $\Icaltwo\lesssim\NN^{-\frac{1}{2}}$ with $\PP$-probability $1-o(1)$. 
\end{lemma}
\begin{proof}
The proof of the statement follows from Lemma~\ref{lem:Chernozhukov6-2}. 
\end{proof}

A version of the following lemma with not only independent but also identically distributed random variables is presented in~\citet[Lemma 6.2]{Chernozhukov2018} and in~\citet[Theorem 5.1, Corollary 5.1]{Chernozhukov2014}. However, as we subsequently show, their results can be generalized to only requiring independence. 

\begin{lemma}\label{lem:Chernozhukov6-2}(Maximal Inequality: \citet[Lemma 6.2]{Chernozhukov2018}; \citet[Theorem 5.1, Corollary 5.1]{Chernozhukov2014})
Let $\eta\in\TauN$, and consider the function class $\Fcaleta:= \{\psi_j(\cdot; \theta,\eta)\colon j\in\indset{d+1+q^2}, \theta\in\Theta\}$.
Suppose that $\Feta\ge \sup_{f\in\Fcaleta} |f|$ is a measurable envelope for $\Fcaleta$ with $\normP{\Feta}{p}<\infty$. 
Let $\kk\in\indset{\KK}$, and let $M := \max_{i\in\Ik} \Feta(\Si)$. 
Let $\tau^2>0$ be a positive constant satisfying $\sup_{f\in\Fcaleta}\normP{f}{2}^2\le \tau^2\le \normP{\Feta}{2}^2<\infty$, where we write $\normP{\phi}{2}^2 = \frac{1}{\normone{\Ik}}\sum_{i\in\Ik}\EP[\phi^2(\Si)]$ for functions $\phi$. Suppose there exist constants $a\ge e$ 
and $v\ge 1$ such that for all $0<\varepsilon\le 1$,
\begin{equation}\label{eq:Chernozhukov6-2:claim1}
	\log\sup_Q N(\varepsilon \norm{\Feta}_{Q, 2}, \Fcaleta, \norm{\cdot}_{Q, 2})
	\le v\log(a/\varepsilon)
\end{equation}
holds, where $Q$ runs over the class $\{\normone{\Ik}^{-1}\sum_{i\in\Ik}Q_i\colon Q_i \text{ a probability measure}\}$ of measures.
Consider the empirical process 
\begin{displaymath}
	\Gnk[\psi(\S)] := \frac{1}{\sqrt{\normone{\Ik}}}\sum_{i\in\Ik} \big( \psi(\Si) - \EP[\psi(\Si)] \big). 
\end{displaymath}
Then, we have
\begin{equation}\label{eq:Chernozhukov6-2:claim2}
	\EP[\norm{\Gnk}_{\Fcaleta}] \le C\cdot \bigg( \sqrt{v\tau^2\log\big(a \normP{\Feta}{2}\tau^{-1}\big)} + \frac{v\normP{M}{2}}{\sqrt{\normone{\Ik}}}\log\big(a \normP{\Feta}{2}\tau^{-1} \big)\bigg),
\end{equation}
where $C$ is an absolute constant. Moreover, for every $t\ge 1$, with probability $>1-t^{-\frac{p}{2}}$, we have 
\begin{equation}\label{eq:Chernozhukov6-2:claim3}
	\norm{\Gnk}_{\Fcaleta}
	\le (1+\alpha)\EP[\norm{\Gnk}_{\Fcaleta}] + C(p) \Big( (\tau + \normone{\Ik}^{-\frac{1}{2}}\normP{M}{p})\sqrt{t} + \alpha^{-1} \normone{\Ik}^{-\frac{1}{2}} \normP{M}{2}t\Big)
\end{equation}
for all $\alpha > 0$, where $C(p)>0$ is a constant depending only on $p$. In particular, setting $a\ge \normone{\Ik}$ and $t=\log(\normone{\Ik})$, with probability $>1-c\cdot(\log(\normone{\Ik}))^{-1}$ for some constant $c$, we have
\begin{displaymath}
	\norm{\Gnk}_{\Fcaleta}
	\le C(p,c)\bigg( \tau\sqrt{v\log\big(a \normP{\Feta}{2}\tau^{-1}\big)}  + \frac{v\normP{M}{2}}{\sqrt{\normone{\Ik}}}\log\big(a \normP{\Feta}{2}\tau^{-1} \big) \bigg),
\end{displaymath}
where $\normP{M}{p}\le \normone{\Ik}^{\frac{1}{p}}\normP{\Feta}{p}$ and $C(p,c)>0$ is a constant depending only on $p$ and $c$. 
\end{lemma}
\begin{proof}
	
	Observe that an envelope $\Feta$ as described in the lemma exists due to Lemma~\ref{lem:exampleVaart}. 
	Consequently, statement~\eqref{eq:Chernozhukov6-2:claim1} holds with $a=\mathrm{diam}(\Theta)$ due to~\citet[Example 19.7]{Vaart1998} and due to Lemma~\ref{lem:exampleVaart}.
	\cite{Liu2020} proceed similarly to establish a similar claim. 
	The proof of~\citet[Corollary 5.1]{Chernozhukov2014} can be adapted to verify statement~\eqref{eq:Chernozhukov6-2:claim2}, and 
	the proof of~\citet[Theorem 5.1]{Chernozhukov2014} can be adapted to show statement~\eqref{eq:Chernozhukov6-2:claim3}. 
	Adaptations are required because these two results 
	are stated for \iid\ data. Our grouped data $\{\Si\}_{i\in\indset{\NN}}$ is groupwise independent, but not identically distributed because a different number of observations may be available for the different groups $i\in\indset{\NN}$. 
	Subsequently, we describe these adaptations. 
	
	The proof of~\citet[Theorem 5.1]{Chernozhukov2014} is based on~\citet[Theorem 12]{Boucheron2005}. The latter result is an inequality for functions of independent random variables and does not require identically distributed variables. Thus, statement~\eqref{eq:Chernozhukov6-2:claim3} is established in our setting. 
	
	Also the proof of~\citet[Theorem 5.2]{Chernozhukov2014} only requires independent but not necessarily identically distributed random variables. Hence, the Corollary 5.1 of Theorem 5.2 in~\citet{Chernozhukov2014} 
	remains to hold in our setting, and thus statement~\eqref{eq:Chernozhukov6-2:claim2} is established as well. 
\end{proof}

\subsection{Asymptotic Distribution of $\hbeta$}\label{sect:asymptoticDistribution}

\begin{proof}[Proof of Theorem~\ref{thm:asymptoticGaussian}]
Fix a sequence $\{\PPN\}_{\NN\ge 1}$ of probability measures such that $\PPN\in\PcalN$ for all $\NN\ge 1$. Because this sequence is chosen arbitrarily, it suffices to show that~\eqref{eq:thmEquation} holds along $\{\PPN\}_{\NN\ge 1}$
to infer that it holds uniformly over $\PP\in\PcalN$. 

Recall the notations $\hbRkXi=\Xi-\hmX^{\Ikc}(\Wi)$ and $\hbRkYi=\Yi-\hmY^{\Ikc}(\Wi)$ for $i\in\indset{\NN}$. 
Observe that the estimator $\hbeta$ in~\eqref{eq:betahat} can alternatively be represented by
\begin{displaymath}
	\hbeta = \frac{1}{\KK}\sum_{\kk=1}^{\KK} \bigg(\argmin_{\beta} \frac{1}{\nnktot}\sum_{i\in\Ik} \big( \hbRkYi - \hbRkXi\beta \big)^T\hVik^{-1}\big( \hbRkYi - \hbRkXi\beta \big)\bigg)
\end{displaymath}
for $\hVik := \Zi\hSigmak\Zi^T + \oneni$ because the Gaussian likelihood decouples. In particular, $\hbeta$ has a generalized least squares representation. 
Observe furthermore that we have
\begin{equation}\label{eq:gauss0}
	\sqrt{\NNtot}(\hbeta-\betazero)
	= \frac{1}{\KK}\sum_{\kk=1}^{\KK} \Big( \frac{1}{\nnktot}\sum_{i\in\Ik} (\hbRkXi)^T\hVik^{-1}\hbRkXi \Big)^{-1}  \frac{\sqrt{\NNtot}}{\nnktot} \sum_{i\in\Ik} (\hbRkXi)^T\hVik^{-1} \big( \hbRkYi-\hbRkXi\betazero \big).
\end{equation}
Let $\kk\in\indset{\KK}$. We have
\begin{equation}\label{eq:gauss1}
	\begin{array}{rl}
		& \frac{\sqrt{\NNtot}}{\nnktot} \sum_{i\in\Ik} (\hbRkXi)^T\hVik^{-1} \big( \hbRkYi-\hbRkXi\betazero \big)\\
		=&\frac{\sqrt{\NNtot}}{\nnktot} \sum_{i\in\Ik}  (\hbRkXi)^T\Vizero^{-1} \big( \hbRkYi-\hbRkXi\betazero \big) 
		+ \frac{\sqrt{\NNtot}}{\nnktot} \sum_{i\in\Ik}  (\hbRkXi)^T\big(\hVik^{-1}-\Vizero^{-1}\big) \big( \hbRkYi-\hbRkXi\betazero \big).
	\end{array}
\end{equation}
We analyze the two terms in the above decomposition~\eqref{eq:gauss1} individually. We start with the second term. For $i\in\indset{\NN}$, $\eta\in\TauN$, and $\Sigma$ from $\Theta$, define the function 
\begin{equation}\label{eq:psiLoss}
	\scoretest(\Si;\Sigma,\eta):=\big(\Xi-\mX(\Wi)\big)^T(\Zi\Sigma\Zi^T+\oneni)^{-1}\Big(\Yi-\mY(\Wi) - \big(\Xi-\mX(\Wi)\big)\betazero\Big). 
\end{equation}
We have 
\begin{equation}\label{eq:gauss2}
	\begin{array}{rl}
		&\frac{\sqrt{\NNtot}}{\nnktot}\sum_{i\in\Ik}  (\hbRkXi)^T\big(\hVik^{-1}-\Vizero^{-1}\big) \big( \hbRkYi-\hbRkXi\betazero \big) \\
		=& \sqrt{\NNtot}\Enkbig{\scoretest(\S;\hSigmak,\hetaIkc) - \scoretest(\S; \Sigmazero,\hetaIkc)}\\
		=& \sqrt{\NNtot}\Enkbig{\scoretest(\S;\hSigmak,\hetaIkc) - \scoretest(\S; \Sigmazero,\etazero)} - \sqrt{\NNtot}\Enkbig{\scoretest(\S;\Sigmazero,\hetaIkc) - \scoretest(\S; \Sigmazero,\etazero)}.
	\end{array}
\end{equation}
Next, we analyze the two terms in~\eqref{eq:gauss2}. The second term is of order
\begin{equation}\label{eq:gauss3}
	\normbig{\sqrt{\NNtot}\Enkbig{\scoretest(\S;\Sigmazero,\hetaIkc) - \scoretest(\S; \Sigmazero,\etazero)}} = o_{\PPN}(1)
\end{equation}
by Lemma~\ref{lem:empiricalSumsSqrtN}. The first term in~\eqref{eq:gauss2} is bounded by
\begin{displaymath}
	\begin{array}{rl}
		& \normbig{\sqrt{\NNtot}\Enkbig{\scoretest(\S;\hSigmak,\hetaIkc) - \scoretest(\S; \Sigmazero,\etazero)}}\\
		\le& \sup_{\norm{\Sigma-\Sigmazero}\le\deltaN} \normbig{ \sqrt{\NNtot}\Enkbig{\scoretest(\S;\Sigma,\hetaIkc) - \scoretest(\S; \Sigmazero,\etazero)}}.
	\end{array}
\end{displaymath}
with $\PPN$-probability $1-o(1)$ due to Lemma~\ref{lem:consistency} because we have $\deltaN^2\le\deltaN$ for $\NN$ large enough.
Let $\Sigma$ be from $\Theta$ with $\norm{\Sigma-\Sigmazero}\le\deltaN$. 
With $\PPN$-probability $1-o(1)$, we have
\begin{equation}\label{eq:gauss7}
	\sqrt{\NNtot} \Enkbig{\scoretest(\S;\Sigma,\hetaIkc) - \scoretest(\S; \Sigmazero,\etazero)} \lesssim \deltaN
\end{equation}
by Lemma~\ref{lem:empProcPsi}. 
Consequently, the second term in~\eqref{eq:gauss1} is of order $o_{\PPN}(1)$ due to~\eqref{eq:gauss2}, \eqref{eq:gauss3}, and~\eqref{eq:gauss7}. 
Subsequently, we analyze the first term in~\eqref{eq:gauss1}. 
By Lemma~\ref{lem:empiricalSumsSqrtN}, we have
\begin{displaymath}
	  \frac{\sqrt{\NNtot}}{\nnktot} \sum_{i\in\Ik}  (\hbRkXi)^T\Vizero^{-1} \big( \hbRkYi-\hbRkXi\betazero \big)
	 =   \frac{\sqrt{\NNtot}}{\nnktot} \sum_{i\in\Ik}  \bRkXi^T\Vizero^{-1} \big( \bRkYi-\bRkXi\betazero \big) + o_{\PPN}(1).
\end{displaymath}
Denote by 
\begin{displaymath}
	\Ti := \EPNbigg{\bRkXi^T\Vizero^{-1} \big( \bRkYi-\bRkXi\betazero \big)\Big(\bRkXi^T\Vizero^{-1} \big( \bRkYi-\bRkXi\betazero \big)\Big)^T}. 
\end{displaymath}
We have 
\begin{equation}\label{eq:simplifyTNi}
	\Ti = \EPNbig{\bRkXi^T\Vizero^{-1}\bRkXi}
\end{equation}
due to Assumption~\ref{assumpt:D4-2}. 
Furthermore, recall $\Tbar = \frac{1}{\NNtot}\sum_{i=1}^{\NN}\Ti$ from Assumption~\ref{assumpt:regularity7}. Due to  Assumption~\ref{assumpt:regularity3-2}, the singular values of the matrices $\Ti$, $i\in\indset{\NN}$ are uniformly bounded away from $0$ by $\clambdamin>0$. Thus, the smallest eigenvalue $\nuN^2$ of $\Tbar$ satisfies 
\begin{equation}\label{eq:gauss11}
		\nuN^2 \ge \frac{1}{\NNtot}\sum_{i=1}^{\NN}\lambdamin\Big(\Ti\Big) \ge \frac{1}{\nmax}\clambdamin>0
\end{equation}
because we have $\NNT\le \NN\nmax$ with $\nmax<\infty$. 
Next, we verify the Lindeberg condition. 
Due to the Cauchy-Schwarz inequality,  
Markov's inequality, H\"{o}lder's inequality, and~\eqref{eq:gauss11}, 
we have
\begin{displaymath}
	\begin{array}{rl}
		&\frac{1}{\NNtot\nuN^2} \sum_{i=1}^{\NN} 
		\EPNbigg{\normbig{\bRkXi^T\Vizero^{-1} \big( \bRkYi-\bRkXi\betazero \big)}^2 \one_{\big\{\normbig{\bRkXi^T\Vizero^{-1} ( \bRkYi-\bRkXi\betazero )}^2\ge\eps\NNtot\nuN^2\big\}}}\\
		\le & \frac{1}{\NNtot\nuN^2} \sum_{i=1}^{\NN} 
		\normPNbig{\bRkXi^T\Vizero^{-1} \big( \bRkYi-\bRkXi\betazero \big)}{4}^2
		\sqrt{\PPN\big(\normbig{\bRkXi^T\Vizero^{-1} ( \bRkYi-\bRkXi\betazero )}^2\ge\eps\NNtot\nuN^2\big)}\\
		\le & \frac{1}{\NNtot\nuN^2} \sum_{i=1}^{\NN} 
		\normPNbig{\bRkXi^T\Vizero^{-1} \big( \bRkYi-\bRkXi\betazero \big)}{4}^2\normPNbig{\bRkXi^T\Vizero^{-1} \big( \bRkYi-\bRkXi\betazero \big)}{2}
		\sqrt{\frac{1}{\eps\NNtot\nuN^2}}\\
		\le & \frac{1}{\NNtot\nuN^2} \sum_{i=1}^{\NN} 
		\normPNbig{\bRkXi^T\Vizero^{-1} \big( \bRkYi-\bRkXi\betazero \big)}{4}^3
		\sqrt{\frac{1}{\eps\NNtot\nuN^2}}\\
		\lesssim& \sqrt{\frac{1}{\eps\NNtot}} \stackrel{\NN\rightarrow\infty}{\longrightarrow} 0
	\end{array}
\end{displaymath}
for $\eps > 0$
by Assumptions~\ref{assumpt:regularity2},~\ref{assumpt:Theta1},~\ref{assumpt:Theta3}, and Lemma~\ref{lem:Emmenegger2021-Lemma-G-7}. 
Consequently, we have
\begin{displaymath}
	(\Tbar)^{-\frac{1}{2}}   \frac{1}{\sqrt{\NNtot}} \sum_{i=1}^{\NN}  \bRkXi^T\Vizero^{-1} \big( \bRkYi-\bRkXi\betazero \big)
	\stackrel{d}{\longrightarrow} \mathcal{N}_d(\bo,\one)
\end{displaymath}
by~\citet[Theorem 6.9.2]{Hansen2017}. Thus, we infer 
\begin{displaymath}
\begin{array}{cl}
	& (\Tbar)^{-\frac{1}{2}} \sqrt{\NNtot} 
	\frac{1}{\KK}\sum_{\kk=1}^{\KK} \frac{1}{\nnktot} \sum_{i\in\Ik}
        \bRkXi^T\Vizero^{-1} \big( \bRkYi-\bRkXi\betazero \big)\\
	= &(\Tbar)^{-\frac{1}{2}}   \frac{1}{\sqrt{\NNtot}} \sum_{i=1}^{\NN}  \bRkXi^T\Vizero^{-1} \big( \bRkYi-\bRkXi\betazero \big) + o_{\PPN}(1)
	\end{array}
\end{displaymath}
due to $\nnktot = \frac{\NNT}{\KK}=o(1)$. 

Finally, the term $\frac{1}{\nnktot}\sum_{i\in\Ik}(\hbRkXi)^T\hVik^{-1}\hbRkXi$ in~\eqref{eq:gauss0} equals $\Tzero + o_{\PPN}(1)$ 
due to Lemma~\ref{lem:multiplierMatrix}.
Therefore, we have
\begin{displaymath}
	\sqrt{\NNtot}\Tzero^{\frac{1}{2}}(\hbeta-\betazero)
	= 
	(\Tbar)^{-\frac{1}{2}} \frac{1}{\sqrt{\NNtot}} \sum_{i=1}^{\NN}  \bRkXi^T\Vizero^{-1} \big( \bRkYi-\bRkXi\betazero \big)
	+ o_{\PPN}(1)
	\stackrel{d}{\longrightarrow} \mathcal{N}_d(\bo,\one).
\end{displaymath}
\end{proof}

\begin{lemma}\label{lem:empiricalSumsSqrtN}
	Let $\kk\in\indset{\KK}$. For $i\in\indset{\NN}$ and $\eta\in\TauN$, consider the function 
	\begin{displaymath}
	\scoretest(\Si;\Sigmazero,\eta)=\big(\Xi-\mX(\Wi)\big)^T(\Zi\Sigmazero\Zi^T+\oneni)^{-1}\Big(\Yi-\mY(\Wi) - \big(\Xi-\mX(\Wi)\big)\betazero\Big)
	\end{displaymath}
	 as in~\eqref{eq:psiLoss}, but where we consider $\Sigmazero$ instead of general $\Sigma$ from $\Theta$. 
	We have 
	\begin{displaymath}
		\normbigg{\frac{\sqrt{\NNtot}}{\nnktot} \sum_{i\in\Ik}\scoretest(\Si;\Sigmazero,\hetaIkc)
		- \frac{\sqrt{\NNtot}}{\nnktot}\sum_{i\in\Ik}\scoretest(\Si;\Sigmazero,\etazero)}
		=O_{\PP}(\deltaN).
	\end{displaymath}
\end{lemma}
\begin{proof}[Proof of Lemma~\ref{lem:empiricalSumsSqrtN}]  A similar proof that is  modified from~\citet{Chernozhukov2018} is presented in~\citet[Lemma G.16]{Emmenegger2021}. 
For notational simplicity, we omit the argument $\Sigmazero$ in $\scoretest$ and write $\scoretest(\Si;\eta)$ instead of $\scoretest(\Si;\Sigmazero,\eta)$. 
 By the triangle inequality, we have
\begin{displaymath}
	\begin{array}{cl}
		&\normBig{\frac{\sqrt{\NNtot}}{\nnktot} \sum_{i\in\Ik}\scoretest(\Si;\hetaIkc)
		- \frac{\sqrt{\NNtot}}{\nnktot} \sum_{i\in\Ik}\scoretest(\Si;\etazero)}\\
		=& \Big\lVert\frac{\sqrt{\NNtot}}{\nnktot} \sum_{i\in\Ik}\big( \scoretest(\Si; \hetaIkc) -\int \scoretest(\si; \hetaIkc)\,\mathrm{d}\PP(\si) \big) 
	 - 
		\frac{\sqrt{\NNtot}}{\nnktot} \sum_{i\in\Ik}\big( \scoretest(\Si; \etazero) -\int \scoretest(\si; \etazero)\,\mathrm{d}\PP(\si) \big) \\
		&\quad\quad+ \sqrt{\NNtot}\frac{1}{\nnktot}\sum_{i\in\Ik}\int \big(\scoretest(\si; \hetaIkc) - \scoretest(\si; \etazero)\big)\,\mathrm{d}\PP(\si)\Big\rVert\\
		\le &\Icalone + \sqrt{\NNtot}\Icaltwo, 
	\end{array}
\end{displaymath}
where $\Icalone:= \norm{M}$
for
\begin{displaymath}
	M :=\frac{\sqrt{\NNtot}}{\nnktot}\sum_{i\in\Ik}\bigg( \scoretest(\Si; \hetaIkc) -\int \scoretest(\si; \hetaIkc)\,\mathrm{d}\PP(\si) \bigg)  - 
		\frac{\sqrt{\NNtot}}{\nnktot} \sum_{i\in\Ik}\bigg( \scoretest(\Si; \etazero) -\int \scoretest(\si; \etazero)\,\mathrm{d}\PP(\si) \bigg),
\end{displaymath}
and where
\begin{displaymath}
	\Icaltwo:=\normbigg{\frac{1}{\nnktot}\sum_{i\in\Ik}\int \big(\scoretest(\si; \hetaIkc) - \scoretest(\si; \etazero)\big)\,\mathrm{d}\PP(\si)}. 
\end{displaymath}
Subsequently, we bound the two terms $\Icalone$ and $\Icaltwo$ individually. 
First, we bound $\Icalone$. Because the dimensions $d$ of $\betazero$ and $q$ of the random effects model are fixed, it is sufficient to bound one entry of the $d$-dimensional column vector $M$. Let $t\in\indset{d}$. On the event $\EpsN$ that holds with $\PP$-probability $1-o(1)$, we have
\begin{equation}\label{eq:rNp1}
	\begin{array}{cl}
		&\EPbig{\norm{M_{t}}^2\big|\SIkc}\\
		=& \frac{\NNtot}{\nnktot^2}\sum_{i\in\Ik}\EP\big[\normone{\scoretest_{t}(\Si; \hetaIkc) -  \scoretest_{t}(\Si; \etazero)}^2\big|\SIkc\big] \\
		&\quad+  \frac{\NNtot}{\nnktot^2} \sum_{i,j\in\Ik,i\neq j}\EP\big[\big( \scoretest_{t}(\Si; \hetaIkc) -  \scoretest_{t}(\Si; \etazero)\big)\big(\scoretest_{t}(\Sj; \hetaIkc) -  \scoretest_{t}(\Sj; \etazero)\big)\big|\SIkc\big] \\
		&\quad -  \frac{2\NNtot}{\nnktot^2}\sum_{i\in\Ik}\EP\big[ \scoretest_{t}(\Si; \hetaIkc) -  \scoretest_{t}(\Si; \etazero)\big|\SIkc\big]\\
		&\quad\quad\quad\quad\quad\quad\cdot \sum_{j\in\Ik}\EP\big[\scoretest_{t}(\Sj; \hetaIkc) - \scoretest_{t}(\Sj; \etazero)\big|\SIkc\big]\\
		&\quad +  \frac{\NNtot}{\nnktot^2}\sum_{i\in\Ik}\EP\big[\scoretest_t(\Si;\hetaIkc)-\scoretest_t(S_i;\etazero)\big| \SIkc\big]^2\\
		&\quad 
		+  \frac{\NNtot}{\nnktot^2}\sum_{i,j\in\Ik,i\neq j} \EP\big[ \scoretest_t(\Si;\hetaIkc)-\scoretest_t(\Si;\etazero) \big| \SIkc \big]
		\EP\big[ \scoretest_t(\Sj;\hetaIkc)-\scoretest_t(\Sj;\etazero) \big| \SIkc \big]\\
		\le &  \frac{\NNtot}{\nnktot^2}\sum_{i\in\Ik}\EP\big[\normone{\scoretest_{t}(\Si; \hetaIkc) -  \scoretest_{t}(S_i; \etazero)}^2\big|\SIkc\big]\\
		\le & \sup_{\eta\in\TauN} \frac{\NNtot}{\nnktot^2}\sum_{i\in\Ik}\EP\big[\norm{ \scoretest(\Si; \eta) -  \scoretest(\Si; \etazero)}^2\big]
	\end{array}
\end{equation}
because $\Si$ and $\Sj$ are independent for $i\neq j$. 
Due to Lemma~\ref{lem:boundPsiEtaEtazero}, we have $\EP[\Icalone^2|\SIkc] \lesssim \deltaN^4\le \delta^2$ for $\NN$ large enough because $\frac{\NNtot}{\nnktot}$ is of order $O(1)$ by assumption. 
Thus, we infer $\Icalone = O_{\PP}(\deltaN)$
by Lemma~\ref{lem:ChernozhukovLemma}. 
Subsequently, we bound $\Icaltwo$. Let $i\in\Ik$. For $r\in [0,1]$, we introduce the function
\begin{displaymath}
	f_k(r) := \frac{1}{\nnktot}\sum_{i\in\Ik}\Big(\EPbig{\scoretest\big(\Si; \etazero + r(\hetaIkc-\etazero)\big)\big| \SIkc} - \EP[\scoretest(\Si; \etazero)]\Big).
\end{displaymath}
Observe that $\Icaltwo= \norm{f_k(1)}$ holds. 
We apply a Taylor expansion to this function and obtain
\begin{displaymath}
	f_k(1)=f_k(0)+f_k'(0) + \frac{1}{2}f_k''(\tilde r)
\end{displaymath}
for some $\tilde r\in (0,1)$. We have 
\begin{displaymath}
	f_k(0) =\frac{1}{\nnktot}\sum_{i\in\Ik}\Big( \EPbig{\scoretest(\Si; \etazero)\big|\SIkc}- \EP[\scoretest(\Si; \etazero)] \Big)= \bo. 
\end{displaymath}
Furthermore, the score $\scoretest$ satisfies the Neyman orthogonality property $f'_k(0)=\bo$
on the event $\EpsN$ that holds with $\PP$-probability $1-o(1)$
because we have for all $i\in\Ik$ and $\eta\in\TauN$ that 
\begin{equation}\label{eq:NeymanOrth1}
	\begin{array}{cl}
		&\frac{\partial}{\partial r}\Big|_{r=0} \EP\big[\scoretest\big(\Si;\etazero+r(\eta-\etazero)\big)\big]\\
			=&\frac{\partial}{\partial r}\Big|_{r=0}  \EP\bigg[ \Big(\Xi-\mX^0(\Wi)-r\big(\mX(\Wi)-\mX^0(\Wi)\big)\Big)^T(\Zi\Sigmazero\Zi^T+\oneni)^{-1}\\
			&\quad\quad\quad\cdot
			\bigg(\Yi-\mY^0(\Wi)-r\big(\mY(\Wi)-\mY^0(\Wi)\big)\\
			&\quad\quad\quad\quad\quad-\Big( \Xi-\mX^0(\Wi)-r\big(\mX(\Wi)-\mX^0(\Wi)\big) \Big)\betazero\bigg) \bigg]\\
			=&  \EP\Big[ -\big(\mX(\Wi)-\mX^0(\Wi)\big)^T(\Zi\Sigmazero\Zi^T+\oneni)^{-1}\Big(\Yi-\mY^0(\Wi)-\big( \Xi-\mX^0(\Wi) \big)\betazero\Big) \\
			&\quad\quad\quad\quad -
			\big(\Xi-\mX^0(\Wi)\big)(\Zi\Sigmazero\Zi^T+\oneni)^{-1}\Big( \mY(\Wi)-\mY^0(\Wi)-\big( \mX(\Wi)-\mX^0(\Wi) \big)\betazero\Big)  \Big]\\
			=& \bo
	\end{array}
\end{equation}
holds because we can apply the tower property to condition on $\Wi$ inside the above expectations, and because $\mX^0$ and $\mY^0$ are the true conditional expectations. 
Moreover, we have
\begin{displaymath}
	\begin{array}{cl}
		&\frac{\partial^2}{\partial r^2} \EP\big[\scoretest\big(\Si;\etazero+r(\eta-\etazero)\big)\big]\\
		=& 2\EP\bigg[ \big(\mX(\Wi)-\mX^0(\Wi)\big)^T(\Zi\Sigmazero\Zi^T+\oneni)^{-1}\Big( \mY(\Wi)-\mY^0(\Wi)-\big(\mX(\Wi)-\mX^0(\Wi)\big)\betazero \Big) \bigg]
		\end{array}
\end{displaymath}
for all $i\in\Ik$. 
On the event $\EpsN$ that holds with $\PP$-probability $1-o(1)$, we  have
\begin{displaymath}
	\norm{f_k''(\tilde r)} \le \sup_{r\in (0,1)} \norm{f_k''(r)} \lesssim \deltaN\NN^{-\frac{1}{2}}
\end{displaymath}
by Lemma~\ref{lem:boundLambdaN}. 
Therefore, we conclude 
\begin{displaymath}
	\normBig{\frac{\sqrt{\NNtot}}{\nnktot} \sum_{i\in\Ik}\scoretest(\Si;\hetaIkc)
		- \frac{\sqrt{\NNtot}}{\nnktot} \sum_{i\in\Ik}\scoretest(\Si;\etazero)}
		\le \Icalone + \sqrt{\NNtot}\Icaltwo
		= O_{\PP}(\deltaN). 
\end{displaymath}
\end{proof}

\begin{lemma}\label{lem:boundPsiEtaEtazero}
	We have
	\begin{displaymath}
		 \sup_{\eta\in\TauN}\frac{1}{\nnktot}\sum_{i\in\Ik}\EP\big[\norm{ \scoretest(\Si; \Sigmazero, \eta) -  \scoretest(\Si; \Sigmazero, \etazero)}^2\big] \lesssim \deltaN^4.
	\end{displaymath}
\end{lemma}
\begin{proof}[Proof of Lemma~\ref{lem:boundPsiEtaEtazero}]
A similar proof that is  modified from~\citet{Chernozhukov2018} is presented in~\citet[Lemma G.15, Lemma G.16]{Emmenegger2021}. 
For notational simplicity, we omit the argument $\Sigmazero$ in $\scoretest$ and write $\scoretest(\Si;\eta)$ instead of $\scoretest(\Si;\Sigmazero,\eta)$. 
Recall the notation $\Vzeroi = \Zi\Sigmazero\Zi^T+\oneni$ for $i\in\indset{\NN}$. 
	Because we have $\sup_{i\in\indset{\NN}}\norm{\Vzeroi^{-1}}\le\CboundV$ by Assumption~\ref{assumpt:Theta3}, we have
	\begin{equation}\label{eq:rNp1-2}
		\frac{1}{\nnktot}\sum_{i\in\Ik}\EPbig{\normbig{\scoretest(\Si;\eta) - \scoretest(\Si;\etazero)}}\lesssim\deltaN^8
	\end{equation}
        by the triangle inequality, 
	H\"{o}lder's inequality, 
	and because we have for all $i\in\Ik$ that 
	$\nni\le\nmax$, $\normP{\Xi - \mX^0(\Wi)}{2}\le\normP{\Xi}{2}$ by Lemma~\ref{lem:EmmeneggerG10} and similarly for $\Yi$, 
 $\normP{\Xi}{2}$ and $\normP{\Yi}{2}$ are bounded by Assumption~\ref{assumpt:regularity2} and H\"{o}lder's inequality, 
	$(\Zi\Sigmazero\Zi^T + \oneni)^{-1}$ is bounded by Assumption~\ref{assumpt:Theta3},
and $\normP{\etazero-\eta}{2}\le\deltaN^8$ holds by Assumption~\ref{assumpt:DML1}.
	
	Furthermore, we have
\begin{equation}\label{eq:rNp2}
	\begin{array}{rl}
	&\EP\big[\norm{ \scoretest(\Si; \eta) -  \scoretest(\Si; \etazero)}^2\big] \\
	\le &\EP[\norm{ \scoretest(\Si; \eta) -  \scoretest(\Si; \etazero)}]
	+ \EP\big[\norm{ \scoretest(\Si; \eta) -  \scoretest(\Si; \etazero)}^2\one_{\norm{ \scoretest(\Si; \eta) -  \scoretest(\Si; \etazero)}\ge 1}\big],
	\end{array}
\end{equation}
and we have
\begin{equation}\label{eq:rNp3}
	\begin{array}{rl}
		&\EP\big[\norm{ \scoretest(\Si; \eta) -  \scoretest(\Si; \etazero)}^2\one_{\norm{ \scoretest(\Si; \eta) -  \scoretest(\Si; \etazero)}\ge 1}\big]\\
		\le& \normPbig{\scoretest(\Si; \eta) -  \scoretest(\Si; \etazero)}{4}^2
		\sqrt{\PP(\norm{ \scoretest(\Si; \eta) -  \scoretest(\Si; \etazero)}\ge 1)}
	\end{array}
\end{equation}
by H\"{o}lder's inequality. 
Observe that the term 
\begin{equation}\label{eq:rNp4}
	\normPbig{\scoretest(\Si; \eta) -  \scoretest(\Si; \etazero)}{4}^2
\end{equation}
is upper bounded by 
the triangle inequality, 
	H\"{o}lder's inequality, 
	because we have $\nni\le\nmax$, $\normP{\Xi - \mX^0(\Wi)}{p}\lesssim\normP{\Xi}{p}$ by Lemma~\ref{lem:Emmenegger2021-Lemma-G-7} 
	and similarly for $\Yi$, 
$\normP{\Xi}{p}$ and $\normP{\Yi}{p}$ are bounded by Assumption~\ref{assumpt:regularity2}, 
	$(\Zi\Sigmazero\Zi^T + \oneni)^{-1}$ is bounded by Assumption~\ref{assumpt:Theta3},
and $\normP{\etazero-\eta}{p}$ is upper bounded by Assumption~\ref{assumpt:DML1}.
By Markov's inequality, 
we furthermore have
\begin{equation}\label{eq:rNp5}
	\PP(\norm{ \scoretest(\Si; \eta) -  \scoretest(\Si; \etazero)}\ge 1) \le \EP[\norm{ \scoretest(\Si; \eta) -  \scoretest(\Si; \etazero)}] \le\deltaN^8
\end{equation}
due to~\eqref{eq:rNp1-2}. 
Therefore, we have	
\begin{displaymath}
	 \sup_{\eta\in\TauN}\frac{1}{\nnktot}\sum_{i\in\Ik}\EP\big[\norm{ \scoretest(\Si;  \eta) -  \scoretest(\Si;  \etazero)}^2\big] \lesssim \deltaN^8+ \deltaN^4 \lesssim \deltaN^4
\end{displaymath}
for $\NN$ large enough due to~\mbox{\eqref{eq:rNp1-2}--\eqref{eq:rNp5}}. 
\end{proof}

\begin{lemma}\label{lem:boundLambdaN}
	Let $\eta\in\TauN$, and let $i\in\indset{\NN}$. We have
	\begin{displaymath}
		\EP\bigg[ \big(\mX(\Wi)-\mX^0(\Wi)\big)^T\Vzeroi^{-1}\Big( \mY(\Wi)-\mY^0(\Wi)-\big(\mX(\Wi)-\mX^0(\Wi)\big)\betazero \Big) \bigg]\lesssim\deltaN\NN^{-\frac{1}{2}}.
	\end{displaymath}
\end{lemma}
\begin{proof}[Proof of Lemma~\ref{lem:boundLambdaN}]
	The claim follows by applying H\"{o}lder's inequality and the Cauchy-Schwarz inequality because $\sup_{i\in\indset{\NN}}\norm{\Vzeroi^{-1}}$ is upper bounded by Assumption~\ref{assumpt:Theta3}, $\Theta$ is bounded, and  
	\begin{displaymath}
		\normP{\mX(\Wi) - \mX^0(\Wi)}{2} \big(\normP{\mY(\Wi) - \mY^0(\Wi)}{2}+\normP{\mX(\Wi) - \mX^0(\Wi)}{2} \big) \le\deltaN\NN^{-\frac{1}{2}}
	\end{displaymath}
	holds by Assumption~\ref{assumpt:DML1}. 
\end{proof}

\begin{lemma}\label{lem:empProcPsi}
	Let $\Sigma$ from $\Theta$ with $\norm{\theta-\thetazero}\le\deltaN^2$. 
	With $\PP$-probability $1-o(1)$, we have
	\begin{displaymath}
		 \sqrt{\NNtot}\Enkbig{\scoretest(\S;\Sigma,\hetaIkc) - \scoretest(\S; \Sigmazero,\etazero)} \lesssim \deltaN.
	\end{displaymath}
\end{lemma}
\begin{proof}[Proof of Lemma~\ref{lem:empProcPsi}]
Observe that we have
\begin{displaymath}
\begin{array}{cl}
	 &\sqrt{\NNtot}\Enkbig{\scoretest(\S;\Sigma,\hetaIkc) - \scoretest(\S; \Sigmazero,\etazero)} \\
	 = & \sqrt{\frac{\NNtot}{\nnktot}}\sqrt{\frac{\normone{\Ik}}{\nnktot}}\Gnkbig{\scoretest(\S;\Sigma,\hetaIkc) - \scoretest(\S; \Sigmazero,\etazero)}
	 + \sqrt{\NNtot}\EPbig{\Enk{\scoretest(\S;\Sigma,\hetaIkc) - \scoretest(\S; \Sigmazero,\etazero)} \big| \SIkc},
	 \end{array}
\end{displaymath}
where the second summand is bounded by $\deltaN$ due to Lemma~\ref{lem:rateDoubleTaylor}, 
and 
where we recall the empirical process notation
\begin{displaymath}
	\Gnk[\phi(\S)] = \frac{1}{\sqrt{\normone{\Ik}}}\sum_{i\in\Ik}\Big( \phi(\Si) - \int\phi(\si) \,\mathrm{d}\PP(\si)\Big)
\end{displaymath}
for some function $\phi$. 
Consider the function class
\begin{displaymath}
	\Fcaltwo := \big\{ \scoretest_j(\cdot;\Sigma,\hetaIkc) - \scoretest_j(\cdot; \Sigmazero,\etazero) \colon j\in\indset{d}, \norm{\Sigma-\Sigmazero}\le\deltaN^2\big\}.
\end{displaymath}
We have $\sqrt{\frac{\NNtot}{\nnktot}}\sqrt{\frac{\normone{\Ik}}{\nnktot}}=O(1)$ by assumption. Therefore, it suffices to bound 
\begin{displaymath}
	\norm{\Gnk}_{\Fcaltwo} = \sup_{f\in\Fcaltwo}\normone{\Gnk[f]}.
\end{displaymath}
To bound this term, we apply Lemma~\ref{lem:Chernozhukov6-2} conditional on $\SIkc$ to the empirical process $\{\Gnk[f]\colon f\in\Fcaltwo\}$ with the envelope $\Ftwo := \FhetaIkc+\Fetazero$ and $\tau = C\rNp$ for a sufficiently large constant $C$, where $\rNp$ is defined by
\begin{equation}\label{eq:def_rNp}
	\rNp := \sup_{\substack{\eta\in\TauN, \\ \norm{\Sigma-\Sigmazero}\le\deltaN^2}} 
	\normPBig{\frac{1}{\normone{\Ik}}\sum_{i\in\Ik}\scoretest(\Si;\Sigma,\eta)-\scoretest(\Si;\Sigmazero,\etazero)}{2}
\end{equation}
and satisfies $\sup_{f\in\Fcaltwo}\normP{f}{2}\lesssim\rNp$ with $\PP$-probability $1-o(1)$. 
The estimated nuisance parameter $\hetaIkc$ can be treated as fixed if we condition on $\SIkc$. Thus, with $\PP$-probability $1-o(1)$, we have
\begin{equation}\label{eq:supFcaltwo}
	\sup_{f\in\Fcaltwo}\normone{\Gnk[f]} \lesssim \rNp\sqrt{\log\Big(\frac{1}{\rNp}\Big)} + \normone{\Ik}^{-\frac{1}{2}+\frac{1}{p}} \log(\normone{\Ik})
\end{equation}
because $\normP{\Ftwo}{p}=\normP{\FhetaIkc+\Fetazero}{p}$ is finite by the triangle inequality 
and Lemma~\ref{lem:exampleVaart}, because $\Fcaltwo\subset\FcalhetaIkc - \Fcaletazero$, 
and because the uniform covering entropy satisfies
\begin{displaymath}
	\begin{array}{rl}
		& \log\sup_Q N\big(\eps\norm{\FhetaIkc+\Fetazero}_{Q,2}, \FcalhetaIkc-\Fcaletazero, \norm{\cdot}_{Q,2}\big)\\
		\le & \log\sup_Q N\big(\frac{\eps}{2}\norm{\FhetaIkc}_{Q,2}, \FcalhetaIkc, \norm{\cdot}_{Q,2}\big) 
		+ \log\sup_Q N\big(\frac{\eps}{2}\norm{\Fetazero}_{Q,2}, \Fcaletazero, \norm{\cdot}_{Q,2}\big) \\
		\le& 2 v\log\Big(\frac{2a}{\epsilon}\Big)
	\end{array}
\end{displaymath}
for all $0<\eps\le 1$ due to~\citet[Proof of Theorem 3]{Andrews1994} as presented in~\citet{Chernozhukov2018}. 
We have $\rNp\le C\deltaN^2$ for some constant $C$ due to Lemma~\ref{lem:bound_rNp}. 
For $\NN$ large enough, we have $\rNp< 1$. 
The function $\alpha\colon (0,1)\ni x\mapsto x\sqrt{\log (x^{-1})}\in\R$ is non-negative,  increasing for $x$ small enough, and satisfies
	$\lim_{x\rightarrow 0^+} x\sqrt{\log (x^{-1})} = 0$. 
Thus, we have $\alpha(\rNp)=o(1)$ and $\alpha(\rNp)\le \alpha(C\deltaN^2)$ for $\NN$ large enough. Moreover, we have $\alpha(x)\lesssim \sqrt{x}$ for $x\in (0,1)$, so that we infer $\alpha(\rNp)\lesssim \deltaN$. Because we assumed $\normone{\Ik}^{-\frac{1}{2}+\frac{1}{p}}\log(\normone{\Ik})\lesssim\deltaN$, we have $\norm{\Gnk}_{\Fcaltwo}\lesssim\deltaN$ with $\PP$-probability $1-o(1)$ as claimed due to~\eqref{eq:supFcaltwo}.
\end{proof}

\begin{lemma}\label{lem:bound_rNp}
Let $\kk\in\KK$. 
Recall 
\begin{displaymath}
	\rNp = \sup_{\substack{\eta\in\TauN, \\ \norm{\Sigma-\Sigmazero}\le\deltaN^2}} 
	\normPBig{\frac{1}{\normone{\Ik}}\sum_{i\in\Ik}\scoretest(\Si;\Sigma,\eta)-\scoretest(\Si;\Sigmazero,\etazero)}{2}
\end{displaymath}
from~\eqref{eq:def_rNp}. We have $\rNp\lesssim\deltaN^2$. 
\end{lemma}
\begin{proof}[Proof of Lemma~\ref{lem:bound_rNp}]
Let $\eta\in\TauN$, $\Sigma$ from $\Theta$ with $\norm{\Sigma-\Sigmazero}\le\deltaN^2$, and $i\in\indset{\NN}$.
We have
\begin{displaymath}
	\scoretest(\Si; \Sigma,\eta) - \scoretest(\Si,\Sigmazero,\etazero)
	= \scoretest(\Si; \Sigma,\eta) - \scoretest(\Si;\Sigma,\etazero) + \scoretest(\Si; \Sigma,\etazero) - \scoretest(\Si;\Sigmazero,\etazero). 
\end{displaymath} 
Let $t\in\indset{d}$. 
We have
\begin{displaymath}
	\begin{array}{cl}
	& \normP{\Enk{\scoretest_t(\S; \Sigma,\eta) - \scoretest_t(\S;\Sigma,\etazero)}}{2}^2 \\
	=& \frac{1}{\nnktot^2}\sum_{i\in\Ik}\EPbig{(\scoretest_t(\Si;\Sigma,\eta) - \scoretest_t(\Si;\Sigmazero,\etazero))^2} \\
	&\quad + \frac{1}{\nnktot^2}\sum_{i, j\in\Ik, i\neq j}\EPbig{\scoretest_t(\Si;\Sigma,\eta) - \scoretest_t(\Si;\Sigmazero,\etazero)}\EPbig{\scoretest_t(\Sj;\Sigma,\eta) - \scoretest_t(\Sj;\Sigmazero,\etazero)} \\
	\lesssim& \deltaN^4
	\end{array}
\end{displaymath}
due to $\Si\independent\Sj$ for $i\neq j$ and similar arguments as presented 
in the proof of Lemma~\ref{lem:boundPsiEtaEtazero}. Furthermore, we have
\begin{displaymath}
	\normP{\Enk{\scoretest(\Si; \Sigma,\etazero) - \scoretest(\Si;\Sigmazero,\etazero)}}{2}
	\lesssim \deltaN^2
\end{displaymath}
due to the Cauchy-Schwarz inequality, $\norm{\Sigma-\Sigmazero}\le\deltaN^2$, 
because we have $\nni\le\nmax$, $\normP{\Xi - \mX^0(\Wi)}{4}\lesssim\normP{\Xi}{4}$ by Lemma~\ref{lem:Emmenegger2021-Lemma-G-7} and similarly for $\Yi$, 
$\normP{\Xi}{4}$ and $\normP{\Yi}{4}$ are bounded by Assumption~\ref{assumpt:regularity2} and H\"{o}lder's inequality, 
$\Zi$ is bounded by Assumption~\ref{assumpt:regularity6}, 
$\Vi^{-1}=(\Zi\Sigma\Zi^T + \oneni)^{-1}$ is bounded by Assumption~\ref{assumpt:Theta3},
$\normP{\etazero-\eta}{p}\le\CpnormEta$ holds by Assumption~\ref{assumpt:DML1} for $\NN$ large enough, 
and $\Theta$ is bounded by Assumption~\ref{assumpt:Theta1}. 
Consequently, we have $\rNp\lesssim\deltaN^2$ due to the triangle inequality. 
\end{proof}

\begin{lemma}\label{lem:rateDoubleTaylor}
Let $\kk\in\indset{\KK}$. 
For $\Sigma$ belonging to $\Theta$, with $\PP$-probability $1-o(1)$, we have \begin{displaymath}
	\normbig{\sqrt{\NNtot}\EPbig{\Enk{\scoretest(\S;\Sigma,\hetaIkc) - \scoretest(\S; \Sigmazero,\etazero)} \big| \SIkc}}\lesssim\deltaN. 
\end{displaymath}
\end{lemma}
\begin{proof}[Proof of Lemma~\ref{lem:rateDoubleTaylor}]
With $\PP$-probability $1-o(1)$, the machine learning estimator $\hetaIkc$ belongs to the nuisance realization set $\TauN$ due to Assumption~\ref{assumpt:DML2}. Thus, it suffices to show that the claim holds uniformly over $\eta\in\TauN$. 
Consider $\eta\in\TauN$ and $\Sigma$ belonging to $\Theta$. 
For $r\in[0,1]$, consider the function
\begin{displaymath}
	f_k(r) := \EP\big[\scoretest\big(\Si;\Sigmazero + r(\Sigma-\Sigmazero), \etazero + r(\hetaIkc-\etazero)\big)\big| \SIkc\big] - \EP[\scoretest(\Si; \Sigmazero,\etazero)].
\end{displaymath}
We apply a Taylor expansion to this function and obtain
\begin{displaymath}
\begin{array}{rl}
	&\sqrt{\NNtot}\EPbig{\Enk{\scoretest(\S;\Sigma,\hetaIkc) - \scoretest(\S; \Sigmazero,\etazero)} \big| \SIkc}\\
	=& \sqrt{\NNtot}f_k(1) \\
	= &\sqrt{\NNtot}\big(f_k(0) + f'_k(0) + \frac{1}{2}f''_k(\tilde r)\big)
	\end{array}
\end{displaymath}
for some $\tilde r\in(0,1)$. 
We have $f_k(0)=\bo$. Next, we verify the Neyman orthogonality property $f_k'(0)=\bo$ and the second-order condition $f''_k( r)\lesssim\deltaN\NN^{-\frac{1}{2}}$ uniformly over $ r\in(0,1)$, which will conclude the proof.  
We have
\begin{equation}\label{eq:NeymanOrth2}
	\begin{array}{cl}
		&\frac{\partial}{\partial r}\Big|_{r=0} \EP\big[\scoretest\big(\Si;\Sigmazero + r(\Sigma-\Sigmazero),\etazero+r(\eta-\etazero)\big)\big]\\
			=&\frac{\partial}{\partial r}\Big|_{r=0}  \EP\bigg[ \Big(\Xi-\mX^0(\Wi)-r\big(\mX(\Wi)-\mX^0(\Wi)\big)\Big)^T(\Zi\Sigmazero\Zi^T+\oneni + r\Zi(\Sigma-\Sigmazero)\Zi^T)^{-1}\\
			&\quad\quad\quad\cdot
			\bigg(\Yi-\mY^0(\Wi)-r\big(\mY(\Wi)-\mY^0(\Wi)\big)\\
			&\quad\quad\quad\quad\quad-\Big( \Xi-\mX^0(\Wi)-r\big(\mX(\Wi)-\mX^0(\Wi)\big) \Big)\betazero\bigg) \bigg]\\
			=&  \EP\Big[ -\big(\mX(\Wi)-\mX^0(\Wi)\big)^T(\Zi\Sigmazero\Zi^T+\oneni)^{-1}\Big(\Yi-\mY^0(\Wi)-\big( \Xi-\mX^0(\Wi) \big)\betazero\Big) \\
			& \quad\quad\quad\quad + 
			\big(\Xi-\mX^0(\Wi)\big)^T\Big(\frac{\partial}{\partial r}\Big|_{r=0}(\Zi\Sigmazero\Zi^T+\oneni+ r\Zi(\Sigma-\Sigmazero)\Zi^T)^{-1}\Big)\\
			& \quad\quad\quad\quad\quad\quad\cdot\Big( \Yi-\mY^0(\Wi)-\big(\Xi-\mX^0(\Wi) \big)\betazero\Big)\\
			&\quad\quad\quad\quad -
			\big(\Xi-\mX^0(\Wi)\big)^T(\Zi\Sigmazero\Zi^T+\oneni)^{-1}\Big( \mY(\Wi)-\mY^0(\Wi)-\big( \mX(\Wi)-\mX^0(\Wi) \big)\betazero\Big)  \Big)\\
			=&\bo,
	\end{array}
\end{equation}
where we apply the tower property to condition on $\Wi$ inside the above expectation, Assumption~\ref{assumpt:D4-2}, 
and that $\mX^0$ and $\mY^0$ are the true conditional expectations. Thus, we have $f'_k(0)=\bo$. Furthermore, we have
\begin{displaymath}
	\begin{array}{cl}
		&\frac{\partial^2}{\partial r^2} \EP\big[\scoretest\big(\Si;\Sigmazero + r(\Sigma-\Sigmazero),\etazero+r(\eta-\etazero)\big)\big]\\
		=& \EP\Big[
			2 \big(\mX(\Wi)-\mX^0(\Wi)\big)^T(\Zi\Sigmazero\Zi^T+\oneni+ r\Zi(\Sigma-\Sigmazero)\Zi^T)^{-1}\\
			&\quad\quad\cdot\Big(\mY(\Wi)- \mY^0(\Wi) -\big( \mX(\Wi)-\mX^0(\Wi) \big)\betazero\Big)
			\Big]\\
			& \quad + 4r \EP\Big[ 
			\big(\mX(\Wi)-\mX^0(\Wi)\big)^T\Big(\frac{\partial}{\partial r}(\Zi\Sigmazero\Zi^T+\oneni+ r\Zi(\Sigma-\Sigmazero)\Zi^T)^{-1}\Big)\\
			&\quad\quad\cdot\Big(\mY(\Wi)- \mY^0(\Wi) -\big( \mX(\Wi)-\mX^0(\Wi) \big)\betazero\Big)\Big] \\
			&\quad + r^2  \EP\Big[ 
			\big(\mX(\Wi)-\mX^0(\Wi)\big)^T\Big(\frac{\partial^2}{\partial r^2}(\Zi\Sigmazero\Zi^T+\oneni+ r\Zi(\Sigma-\Sigmazero)\Zi^T)^{-1}\Big)\\
			&\quad\quad\cdot\Big(\mY(\Wi)- \mY^0(\Wi) -\big( \mX(\Wi)-\mX^0(\Wi) \big)\betazero\Big)\Big], \\
	\end{array}
\end{displaymath}
where we apply the tower property to condition on $\Wi$ inside the above expectation, Assumption~\ref{assumpt:D4-2}, 
and that $\mX^0$ and $\mY^0$ are the true conditional expectations.
All the above summands are bounded by $\deltaN\NN^{-\frac{1}{2}}$ in the $L^1$-norm due to H\"{o}lder's inequality and Assumptions~\ref{assumpt:regularity2}, \ref{assumpt:Theta1}, \ref{assumpt:Theta3}, and \ref{assumpt:DML1} because for 
$\Ubi:=\mX(\Wi)-\mX^0(\Wi) \in\R^{\nni \times d}$,  $\Vbi := \mY(\Wi)- \mY^0(\Wi) -\Ubi\betazero\in\R^{\nni}$, and a non-random matrix $\Di \in\R^{\nni\times\nni}$ with bounded entries, we have for $j\in\indset{d}$ that
\begin{displaymath}
	\normP{(\Ubi^T)_{j,\cdot} \Di \Vbi}{1} 
	= \normPBig{\sum_{\kappa, \iota=1}^{\nni}(\Ubi^T)_{j,\kappa}(\Di)_{\kappa,\iota}(\Vbi)_{\iota}}{1}
	\le \nni^2 \normP{\Ubi}{2} \normP{\Vbi}{2} \sup_{\kappa,\iota\in\indset{\nni}} \normone{(\Di)_{\kappa, \iota}}
\end{displaymath}
holds due to the triangle inequality and H\"{o}lder's inequality. Because we have $\nni\le\nmax$ uniformly over $i\in\indset{\NN}$, we infer our claim due to
\begin{displaymath}
	\norm{f_k''(\tilde r)}
	\le \sup_{r\in (0,1)} \normBig{\frac{\partial^2}{\partial r^2} \EP\big[\scoretest\big(\Si;\Sigmazero + r(\Sigma-\Sigmazero),\etazero+r(\eta-\etazero)\big)\big]}
	\lesssim \deltaN\NN^{-\frac{1}{2}}.
\end{displaymath}
\end{proof}

\begin{lemma}\label{lem:multiplierMatrix}
Recall the notation $\hVik = \Zi\hSigmak\Zi^T + \oneni$.
We have
\begin{displaymath}
	\frac{1}{\nnktot}\sum_{i\in\Ik}(\hbRkXi)^T\hVik^{-1}\hbRkXi
	= \Tzero + o_{\PP}(1).
\end{displaymath}
\end{lemma}
\begin{proof}[Proof of Lemma~\ref{lem:multiplierMatrix}]
Let us introduce the score function
\begin{displaymath}
	\scoreX(\Si; \Sigma,\eta) := \big(\Xi-\mX(\Wi)\big)^T(\Zi\Sigma\Zi^T+\oneni)^{-1}\big(\Xi-\mX(\Wi)\big)
\end{displaymath}
for $\eta\in\TauN$ and $\Sigma$ from $\Theta$. 
Recall the notation $\Vizero = \Zi\Sigmazero\Zi^T+\one_{\nni}$.
We have 
\begin{equation}\label{eq:decomp}
	\begin{array}{rl}
		&\frac{1}{\nnktot}\sum_{i\in\Ik} \Big( (\hbRkXi)^T \hVik^{-1}\hbRkXi - 
		\EPbig{(\bRkXi)^T\Vizero^{-1}\bRkXi} \Big)\\
		=& \Enkbig{\scoreX(\S;\hSigmak, \hetaIkc) - \EP[\scoreX(\S; \Sigmazero, \etazero) ]}\\
		=&  \Enkbig{\scoreX(\S;\hSigmak, \hetaIkc) - \scoreX(\S;\Sigmazero, \hetaIkc)}
		+ \Enkbig{\scoreX(\S; \Sigmazero,\hetaIkc) - \scoreX(\S; \Sigmazero,\etazero)}\\
		&\quad
		+ \Enkbig{\scoreX(\S; \Sigmazero,\etazero) - \EP[\scoreX(\S; \Sigmazero, \etazero)] }. 
	\end{array}
\end{equation}
The last summand 
$\Enk{\scoreX(\S; \Sigmazero,\etazero) - \EP[\scoreX(\S; \Sigmazero, \etazero)]}$ 
in~\eqref{eq:decomp} is of size $o_{\PP}(1)$ due to Markov's inequality and Assumptions~\ref{assumpt:regularity2},~\ref{assumpt:Theta3}, and~\ref{assumpt:DML1}. 
The second summand 
$\Enk{\scoreX(\S; \Sigmazero,\hetaIkc) - \scoreX(\S; \Sigmazero,\etazero)}$ in~\eqref{eq:decomp} is of size $o_{\PP}(1)$ due to similar arguments as presented in Lemma~\ref{lem:empiricalSumsSqrtN}. This lemma is stated for a slightly different score function that involves $\betazero$, but the proof of this lemma does not depend on $\betazero$. It can be shown that the same arguments are also valid for the score $\scoreX$. The first summand
$ \Enkbig{\scoreX(\S;\hSigmak, \hetaIkc) - \scoreX(\S;\Sigmazero, \hetaIkc)}$
in~\eqref{eq:decomp} is of order $o_{\PP}(1)$. To prove this last claim, 
recall that $\norm{\hthetak-\thetazero}\le\deltaN$ holds with $\PP$-probability $1-o(1)$ due to Lemma~\ref{lem:consistency} and because we have $\deltaN^2\le\deltaN$ for $\NN$ large enough. Consider $\Sigma$ from $\Theta$ 
with $\norm{\Sigma-\Sigmazero}\le\deltaN$, and recall the notation $\Vi = \Zi\Sigma\Zi^T+\one_{\nni}$. 
We have
\begin{equation}\label{eq:Tconsistent1}
	\begin{array}{rl}
		&  \Enkbig{\scoreX(\S;\Sigma, \hetaIkc) - \scoreX(\S;\Sigmazero, \hetaIkc)}\\
		=& \frac{1}{\nnktot}\sum_{i\in\Ik}\big(\Xi - \mX^0(\Wi)\big)^T(\Vi^{-1}-\Vizero^{-1})\big(\Xi-\mX^0(\Wi)\big) \\
		&\quad + \frac{1}{\nnktot}\sum_{i\in\Ik}\big(\Xi - \mX^0(\Wi)\big)^T(\Vi^{-1}-\Vizero^{-1})\big(\mX^0(\Wi)-\hmX^{\Ikc}(\Wi)\big) \\
		&\quad + \frac{1}{\nnktot}\sum_{i\in\Ik}\big(\mX^0(\Wi)-\hmX^{\Ikc}(\Wi)\big)^T(\Vi^{-1}-\Vizero^{-1})\big(\Xi - \mX^0(\Wi)\big) \\
		&\quad + \frac{1}{\nnktot}\sum_{i\in\Ik}\big(\mX^0(\Wi)-\hmX^{\Ikc}(\Wi)\big)^T(\Vi^{-1}-\Vizero^{-1})\big(\mX^0(\Wi)-\hmX^{\Ikc}(\Wi)\big) .
	\end{array}
\end{equation}
The first summand in the decomposition~\eqref{eq:Tconsistent1} is of order $o_{\PP}(1)$ because we have for all $i\in\indset{\NN}$ that
\begin{displaymath}
	\begin{array}{rl}
	&\EPbig{\normbig{\big(\Xi - \mX^0(\Wi)\big)^T(\Vi^{-1}-\Vizero^{-1})\big(\Xi-\mX^0(\Wi)\big)\big) }}\\
	\le & \sup_{i\in\indset{\NN}}\norm{\Vi^{-1}-\Vzeroi^{-1}}\normP{\Xi-\mX^0(\Wi)}{2}^2 \\\lesssim& \deltaN \normP{\Xi-\mX^0(\Wi)}{2}^2 
	\end{array}
\end{displaymath}
holds due to the Cauchy-Schwarz inequality,  H\"{o}lder's inequality, and  Assumption~\ref{assumpt:Theta3}. 
We have $\normP{\Xi-\mX^0(\Wi)}{2}\le\normP{\Xi}{2} <\infty$ due to Lemma~\ref{lem:Emmenegger2021-Lemma-G-10} and Assumption~\ref{assumpt:regularity2}. 
The other summands in~\eqref{eq:Tconsistent1} are of smaller order than the first summand in~\eqref{eq:Tconsistent1} due to Assumption~\ref{assumpt:DML1} and similar computations. 
Therefore, we have
\begin{displaymath}
	\frac{1}{\nnktot}\sum_{i\in\Ik}(\hbRkXi)^T\hVik^{-1}\hbRkXi
	= \frac{1}{\nnktot}\sum_{i\in\Ik} \EPbig{(\bRkXi)^T\Vi^{-1}\bRkXi} + o_{\PP}(1)
	= \Tzero + o_{\PP}(1)
\end{displaymath}
due to Assumption~\ref{assumpt:regularity7}. 
\end{proof}

\section{Stochastic Random Effects Matrices}\label{sec:nonfixedZi}

We considered fixed random effects matrices $\Zi$ in our model~\eqref{eq:PLMM}. 
However, it is also possible to consider stochastic random effects matrices $\Zi$ and to 
include the nonparametric variables $\Wi$ into the random effects matrices. 
In this case, we consider the composite random effects matrices $\Ztili = \h(\Zi,\Wi)$ for some known function $\h$ instead of $\Zi$ in the partially linear mixed-effects model~\eqref{eq:PLMM}. That is, we replace the model~\eqref{eq:PLMM} by the model
\begin{equation}\label{eq:PLMMrandom}
	\Yi = \Xi\betazero + \g(\Wi) + \Ztili\bbi+ \epsi, \quad i\in\indset{\NN}
\end{equation}
with $\Ztili = \h(\Zi,\Wi)$ and $\Zi$ random. We require groupwise independence 
 $\Zi\independent\Zj$ for $i\neq j$ 
of the random effects matrices. 

If $\Zi$ is random, one needs to also condition on it in~\eqref{eq:resNormal}, and we need to assume that the density $p(\Wi, \Xi, \Zi)$ does not depend on $\theta$. 
Furthermore, $\Zi$ needs to be such that the Neyman orthogonality properties~\eqref{eq:NeymanOrth1} and~\eqref{eq:NeymanOrth2} and Equation~\eqref{eq:simplifyTNi} still hold. 
For instance, these equations remain valid if Assumption~\ref{assumpt:D4-2} is replaced by  $(\bbi,\epsi)\independent(\Wi, \Xi)|\Zi$ and $\EP[\bbi|\Zi]=\bo$ and $\EP[\epsi|\Zi]=\bo$ for all $i\in\indset{\NN}$.

Furthermore, the composite random effects matrices $\Ztili$ need to satisfy additional regularity conditions. The Assumptions~\ref{assumpt:regularity6} and~\ref{assumpt:Theta3} need to be adapted as follows. The first option is to adapt Assumption~\ref{assumpt:regularity6} to: there exists a finite real constant $\CboundZi'$ such that $\normP{\Ztili}{\infty}\le\CboundZi'$ holds for all $i\in\indset{\NN}$, where $\normP{\cdot}{\infty}$ denotes the $L^{\infty}(P)$-norm.
Then, Assumption~\ref{assumpt:Theta3} needs to be adapted to: there exists a finite real constant $\CboundV'$ such that we have 
	$\normP{(\Ztili\Sigma\Ztili^T + \oneni )^{-1}}{\infty}\le \CboundV'$.
	for all $i\in\indset{\NN}$ and all $\Sigma$ belonging to $\Theta$.
	
The Assumptions~\ref{assumpt:DML1} and~\ref{assumpt:DML2} formulate the product relationship of the machine learning estimators' convergence rates in terms of the $L^2(P)$-norm. 
The second option is to consider $L^t(P)$-norms with $t\ge 4>2$ in these assumptions instead. Then, it is possible to constrain the $L^p(P)$-norms of $\Ztili$ and $(\Ztili\Sigma\Ztili^T + \oneni )^{-1}$ in Assumptions~\ref{assumpt:regularity6} and~\ref{assumpt:Theta3} instead of their $L^{\infty}(P)$-norm. However, the order $p$, which is specified in Assumption~\ref{assumpt:regularity}, needs to be increased to $p\ge 2^9$ to allow us to bound the terms in the respective proofs by H\"{o}lder's inequality.

\end{appendices}

\end{document}